%% file: article_arxiv.tex
\documentclass[review,onefignum,onetabnum]{siamart171218}

\input{ex_shared}

\nolinenumbers

\usepackage[ulem=normalem,draft]{changes}
\newcommand{\sol}{\textcolor{black}}

\newcommand{\fel}{\textcolor{black}}

\usepackage{accents}
\usepackage{cancel}

\newcommand\munderbar[1]{%
  \underaccent{\bar}{#1}}

\ifpdf
\hypersetup{
  pdftitle={Asymptotic behaviour of cooperative systems: an application to malaria control},
  pdfauthor={Antunes, Aronna, Code\c{c}o}
}
\fi

\begin{document}

\maketitle

% REQUIRED
\begin{abstract}
In this work we propose a model that represents the relation between fish ponds, the mosquito population and the transmission of malaria.
  It has been observed that in the Amazonic region of Acre, in the North of Brazil, fish farming is correlated to the transmission of malaria when carried out in artificial ponds that become breeding sites. Evidence has been found indicating that cleaning the vegetation from the edges of the crop tanks helps to control the size of the mosquito population. 

 We use our model to determine the effective contribution of {fish farming practices on malaria transmission dynamics}. The model consists of a nonlinear system of ordinary differential equations with jumps at the cleaning time, which act as {\em impulsive controls}. We study the asymptotic behaviour of the system in function of the intensity and periodicity of the cleaning, and the value of the parameters. In particular, we state sufficient conditions under which the mosquito population is {eliminated or persists}, and under which the malaria is {eliminated} or becomes endemic. \sol{We prove our conditions by applying results for cooperative systems with concave nonlinearities.}
\end{abstract}

% REQUIRED

% REQUIRED

\section{Introduction}

Malaria is a disease caused by parasites of genus {\em Plasmodium}, transmitted to human beings through the bite of the female {\em Anopheles} mosquito.
Symptoms of malaria range from fever, tiredness, shivers, vomits and headaches in mild cases, to coma or even death in severe cases. It is specially dangerous to pregnant women and small children. Of the five species of {\em Plasmodium}, {\em Plasmodium falciparum} causes the most virulent malaria, and {\em Plasmodium vivax} is the most frequent cause of recurrent malaria {in Brazil} \cite{Carlos2019_ComprehensiveAnalysisMalaria}.

{The Amazon region concentrates} $99.9\%$ of {all} malaria cases reported in {Brazil}, with an annual mean of {150,000 reported cases between 2019 and 2021, with 1,500 hospitalizations and 50 deaths per year}. Even inside the Amazon region, incidence of the disease is concentrated: only 8\% of the municipalities were classified as medium or high risk in 2021 \cite{MinSaude2022_BoletimEpidemiologico}. In the state of Acre there are four counties with high incidence of the disease: Cruzeiro do Sul, M\^ancio Lima, Rodrigues Alves and Tarauacá \cite{Reis2015plos}. {In this region, characterized by vast wetlands, malaria vectors reproduce in the shaded mats of floating vegetation, debris, and along the shaded margins of slowly running creeks \cite{manguin1996characterization}. Humans living close to these areas are historically at higher risk of malaria exposition.} {In the 2000's,} fish farming was stimulated as a sustainable way to develop the local economy, as part of the Brazilian Federal Government's poverty alleviation program. {A large amount of ponds were created in small farms and within backyards, often very close to dwellings and towns.  These new water reservoirs, when not properly managed, accumulated vegetation and debris, and became new breeding sites for {\em Anopheles} mosquitoes, amplifying the risk of malaria transmission \cite{monnerat2014inventory,Reis2015malj}.} Works such as \cite{Reis2015malj} give evidence that fish ponds {were responsible for} an increase in mosquito population, with four times more larvae than natural water bodies. Cases of malaria {were found to be} spatially and temporally correlated with the opening of fish ponds. The most significant predictor of larval {abundance} \cite{Reis2015malj} was percentage of border {with} vegetation. This is attributed to the {fact that vegetation acts as a refuge for the {\em Anopheles}  larvae and pupae, by providing  shade, food, still water and protection from larvivorous fish predation \cite{danis1997aquatic,Reis2015malj,kumar2006larvicidal}.} 

A difference in infestation levels has also been  observed between commercial and non-commercial ponds, {which is associated to the presence of larvivorous fishes in the former} \fel{(see \cite{Reis2015malj}).} % {Removal of border vegetation is an important malaria control measure that depends on the engagement of the farmer community.} 
Fish farmers {tend to} keep their ponds clean of vegetation due to concerns regarding production, {or to avoid the} presence of snakes and other animals.  While there is government incentive for the construction of ponds, there is no incentive for {their} maintenance {with the purpose of malaria protection}. {Moreover, many unproductive ponds are abandoned and there is a lack of policy for land-filling} of ponds that are no longer in use \cite{Reis2015malj,Reis2015plos}.

Due to the impact border vegetation on larval abundance, cleaning practices can be used as a measure of malaria control, {but it depends exclusively on the engagement of the farmer community}. Fish farmers were interviewed for information regarding {their practices,} the cost of border cleaning and the speed of vegetation growth, during an educational event  {\em ``Fish farming without malaria''} organized by Oswaldo Cruz Foundation {in 2018, in Mancio Lima, Acre}. {This survey \cite{alves_codeco_peiter_souza-santos_2019} provides the background information for the mentioned study}.

% Description of the model

 In this article we propose a model that puts together the dynamics of the mosquito population {bred} in fish ponds, the {growth-cut dynamics} of vegetation along the ponds' borders, and the evolution of malaria transmission. We consider that vegetation has a direct impact on the ponds capacity of sheltering aquatic-stage mosquitoes and we further incorporate the possibility of presence of larvae predatory fish. Some simplifying assumptions are made. The system dynamics is considered spatially independent, and homogeneous for each compartment. We assume that vegetation cleaning is synchronized, happens periodically, and its time frame is small when compared to the time frame of the mosquito population and disease dynamics, ({\em i.e.}{,} cleaning is an event that happens in a few hours, and after a number of days it occurs again; while larvae maturation, mosquito lifespan and the spread of malaria through mosquito bites are phenomena that have a daily timescale). Therefore, we model vegetation cleaning as impulses that are applied to the vegetation variable, this is, at the cleaning  \sol{times} we apply a jump to the vegetation with the magnitude of the proportion of vegetation cleaned at that time. \fel{Our model and its analysis were inspired by the work \cite{DumontThuilliez2016} by Dumont and Thuilliez.}
 \sol{For additional references to malaria and malaria control modelling, we refer the reader to \cite{ross1911prevention,Mckenzie2004,smith2004statics,mandal2011mathematical} and references therein.}

\sol{Our main results establish algebraic conditions on the parameters of the model and of the cleaning frequency and intensity that determine the asymptotic behaviour of the mosquito population and the evolution of the disease.
 In order to predict this asymptotic behaviour, we prove and apply an extension of a result by  \sol{Jiang} \cite{jifa1993cooperative}. The latter holds for cooperative systems with concave nonlinearities and states sufficient conditions under which the trajectories either converge to an equilibrium or are attracted to a periodic solution. This article of  \sol{Jiang} \cite{jifa1993cooperative} extends previous works on cooperative systems by Smith \cite{smith1986cooperative} and other similar results for autonomous cooperative systems proved by Hirsch \cite{hirsch1984dynamical,hirsch1985systems}  and  \sol{Jiang} \cite{jifa1990algebraic}. For more historical and technical details on the subject the reader is referred to \cite{jifa1993cooperative}. In particular, in this article (see Appendix \ref{appendix_cooperative}) we extend the algebraic conditions of \cite[Theorem 5.5]{jifa1993cooperative} to dynamics that are merely measurable with respect to the time variable.}

The article is organized as follows.
A detailed description of the model, its components and assumptions are given in Section \ref{modeling:the_model}.
We derive analytical and numerical results from the model. In Section \ref{modeling:analysis} we apply Theorem \ref{Ch3:Main_Theo} to the study of the asymptotic behaviour of our system, establishing sufficient conditions for the mosquito population and the disease to either assume asymptotic periodic behaviour or to be eradicated.
In Section \ref{Ch2:NumSimSection}, we simulate the system numerically in order to illustrate its different possible asymptotic behaviours. We also simulate different values for cleaning periodicity to show the effect of the cleaning frequency on the incidence of malaria. Finally, in the Appendix {\ref{AppK}}, \ref{Ch3:sec2}, \ref{appendix_cooperative} and \ref{appendix_impulsive}, we include a series of definitions and technical results mainly concerning ordinary and impulsive differential equations, and we prove an extension of the main result of \cite{jifa1993cooperative} that holds for less regular right hand-sides.

\subsection{Notation and preliminary definitions}\label{SecNotation} % Section after Introduction

    In this section, we introduce some notation and definitions that will be used throughout the article. In particular, we present elements related to cooperative systems and concave nonlinearities that were employed in \cite{jifa1993cooperative} and that are necessary for establishing our main results.
    
    Given any function $u$ defined in the interval $(t,t+\epsilon_0)$ (resp., $(t-\epsilon_0,t)$), set $u(t^+) \coloneqq \lim_{\epsilon \to 0^+} u(t + \epsilon)$ (resp., $u(t^-) \coloneqq \lim_{\epsilon \to 0^-} u(t - \epsilon)$).
    For $x,y \in \mathbb{R}^n$, we write $x < y$ and $x \leq y$ if the inequalities hold component-wise. We use $x \lneq y$ if $x \leq y$ and $x \neq y$. If $x$ and $y$ are  matrices, then the inequalities should hold entry-wise.  For $x < y$, we define $[x,y] \coloneqq \{z : x\leq z \leq y\}$, $(x,y] \coloneqq \{z : x < z \leq y\}$ and $(x,y) = \{z : x < z < y\}.$
    
   Given an open set $\Omega \subseteq \mathbb{R}\times\mathbb{R}^n$, Our model will be a differential equation of the form
\begin{equation}\label{sec2:eq1}
\dot x = \mathcal{F}(t,x),    
\end{equation}
    where $  \mathcal{F} \colon \Omega \to \mathbb{R}^n$ is differentiable w.r.t. $x$, measurable and $\tau$-periodic w.r.t.  $t$.
   The function $\mathcal{F}$ is called \textit{concave} if $0 < x < y $ implies 
    $$ 
    D_x  \mathcal{F} (t,x) \gneq D_x  \mathcal{F}(t,y), 
    $$ 
    where $ D_x \mathcal{F}(t,x)$ denotes the {\em Jacobian matrix} of $ \mathcal{F}$ with respect to $x$. The system is called {\em cooperative} if 
    $$
    \frac{\partial  \mathcal{F}_i}{\partial x_j} (t,x) \geq 0,\quad \text{for } i \neq j,
    $$ 
    for $(t,x) \in \Omega$. This is equivalent to saying that $ D_x  \mathcal{F} (t,x)$ has non-negative off-diagonal terms.

 Given a Cauchy problem associated to \eqref{sec2:eq1} and with initial condition $x(t_0) = x_0,$ we write ${\bf x} (t;t_0,x_0)$ to denote the value of an associated trajectory at time $t$.
   A {\em Carath\'eodory solution} of a Cauchy problem is an absolutely continuous function $t \mapsto {\bf x} (t;t_0,x_0)$ defined on some interval $[t_0,t_1]$ which satisfies \eqref{sec2:eq1} almost everywhere. We may simply write ${\bf x}(t)$ if both $t_0$ and $x_0$ are implicit from the context.
   
If a set $\mathcal{K} \subset \mathbb{R}^n$ is such that for every $x_0\in \mathcal{K}$, the solution to the Cauchy problem \eqref{sec2:eq1} with $x(t_0) = x_0$ satisfies $x(t) \in \mathcal{K}$ for every $t > t_0$ for which the solution $x$ is defined, then we say that $\mathcal{K}$ is a {\em forward invariant set} for $\mathcal{F}$.  
We say that a trajectory ${\bf x}$  is {\em globally asymptotic attractive} for \eqref{sec2:eq1} if, for every $(t_0,x_0) \in \Omega$, the solution $x$ of \eqref{sec2:eq1} with initial condition $x(t_0)=x_0$ satisfies
\begin{equation}
    \lim_{n\to \infty} \max_{t \in [n,\infty)} |{\bf x}(t;t_0,x_0) - x(t)| = 0.
\end{equation}

   An $n \times n$-matrix $A$ is called {\em cooperative} if its off-diagonal entries are non negative. The Jacobian matrix of a cooperative system is a cooperative matrix.
 Any $k\times k$-submatrix of $A$ formed by deleting $n - k$ rows of $A$, and the corresponding $n - k$ columns is called a {\em principal submatrix} of $A$, and its determinant is called a {\em principal minor} of order $k$. If the $n - k$ removed lines and columns are the last ones, we say that the resulting matrix is a {\em leading principal submatrix}, and its determinant is the {\em leading principal minor}.
 
\sol{Given a positive $\tau$ and a time-dependent $\tau$-periodic essentially bounded function $g\colon [0,+\infty) \to \mathbb{R},$ we set $\overline{g} \coloneqq {\rm ess\,sup}_{[0,\tau)} g (t)$ and $\underline{g} \coloneqq {\rm ess\,inf}_{[0,\tau)} g (t).$}
We extend this notation \sol{to matrix-valued functions} $A\colon [0,+\infty) \to \mathbb{R}^{n\times m}$ in the following sense: if $a_{ij}$ denote the entries of $A$ then set, for $i=1,\dots,n$ and $j=1,\dots,m,$
\[
\overline a_{ij} \coloneqq {\rm ess\,sup}_{[0,\tau)} a_{ij} (t),\quad 
\underline{a}_{ij} \coloneqq {\rm ess\,inf}_{[0,\tau]} a_{ij} (t),
\]
and write $\overline A$ and $\underline A$ for the constant $n\times m$-matrices with entries $\overline a_{i,j}$ and $\underline a_{i,j}$, respectively.

\section{The Model}\label{modeling:the_model}
    
    In this section we introduce the model. More precisely, in Subsection \ref{modeling:the_model:malaria}, we present a model for the interaction between the spread of malaria and the mosquito population dynamics. Subsection \ref{modeling:the_model:vegetation} is dedicated to the model for border vegetation growth and human action, that consists in border cleaning. Finally, Subsection \ref{modeling:the_model:full} joins both systems in order to asses cleaning border vegetation as a control method for malaria.

\subsection{Malaria Model}\label{modeling:the_model:malaria}
    
   A  malaria model will be introduced, joining ideas  from \sol{Dumont {\em et al.} \cite{DumontChiroleuDomerg2008,DumontTchuenche2012,DumontThuilliez2016}} and the classical model by Ross \cite{ross1911prevention}. 
   
   {Let us consider a community of humans living close to fish ponds in a malaria endemic region. 
   They are fish farmers, farm workers, family, neighbors, among others, who are exposed to mosquitoes that breed in these ponds. 
   Vegetation grows at the border of these fish ponds creating a habitat for {\em Anopheles} mosquitoes, while a fraction of ponds contains larvivorous fishes that can partially control the mosquito population. 
   Vegetation cleaning and presence of fishes thus act as tools for mosquito density control.} 
   
   {Malaria dynamics in the human population is described by a Susceptible-Infected-Susceptible model (SIS), which is standard for {\em Plasmodium vivax} malaria.}
   {This disease} has negligible mortality and a high rate of reinfection \cite{Reis2015malj,Reis2015plos}. 
   {For simplicity, this human community is closed, {\em i.e.,} with no demographic rates. 
   They are also considered to be homogeneous in terms of susceptibility and risk exposure. }
    
   % \remove{We combine an SIS model for the disease with a two-stage aquatic-adult population model for the mosquitoes. Both the human and adult mosquito populations are divided in two compartments, namely susceptible and infected. The proportion of susceptible humans is denoted by S, and the proportion of infected  by I. In regard to mosquito population, we consider an aquatic stage L gathering larvae and pupae,  subject to resource competition; and two adult stages, the susceptible adults MS and the infected adults MI.}

    {The vector population is divided into two stages: immatures that inhabit the fish ponds ($L$) and adult mosquitoes ($M$) that emerge from the ponds to feed on humans and eventually become infected.
    Once infected, the mosquito remains transmitting the parasite until death. 
    We assume that mosquito population regulation occurs at the immature stage, due to a density dependent recruitment rate. 
    The literature on the life history of {\em Anopheles} shows that larval crowing can increase mortality, extend development times and produce smaller and more fragile adults, that are less fertile \cite{gimnig2002density, muriu2013larval}.
    To represent the density-dependent effect of larval abundance on recruitment rate in a simplified way, we use a logistic function with a maximum {\em per-capita} recruitment rate $\alpha$ that is achieved when $L = 0$ and which decreases linearly as $L$ increases towards $K$ \cite{mandal2011mathematical}.
    $K$ is the slope of the larval-density dependent recruitment rate, thus, $L/K$ is a measure of how much the recruitment rate is reduced due to the crowding conditions. 
    In the ecological literature, K is often referred to as {\em “carrying capacity”,} but there is considerable controversy about this terminology and its meaning \cite{zhang2021carrying}.
    We assume that $K$ is proportional to the amount of border vegetation, since this constitutes the habitat where the {\em Anopheles} larvae live and compete for resources. The larger the area, the higher is $K$. 
    We further consider that the presence of predatory fishes acts to reduce $K$.  The underlying assumption is that part of the vegetated border is accessible to predatory fishes, limiting the habitat size and increasing crowding conditions \cite{Reis2015malj, orr1989experimental}. Additional comments on the parameter $K$ have been included in Appendix \ref{AppK}.}
    
    %\remove{Assuming initially that the maximum capacity for the aquatic-stage mosquitoes L in  fish ponds is a positive constant K,} 
    The evolution of the state variables is given by the following system of differential equations:
\begin{subequations}\label{Ch2:Sec2:sys1}
\begin{align}
     {\dot S} &= -\beta_{vh} S M_I + \kappa I, \\
     {\dot I} &= \beta_{vh} S M_I - \kappa I, \\
     {\dot M_S} &= \nu L - {\beta_{hv} I M_S} - \mu_M M_S,\\
     {\dot M_I} &= {\beta_{hv} I M_S} - \mu_M M_I,\\
     {\dot L} &= \alpha (M_I + M_S)\left(1 - \frac{L}{K} \right) - (\nu + \mu_L)L,
\end{align}
\end{subequations}
where the meaning of the involved parameters is given in Table \ref{model:table1}.
    {As already mentioned, we} assume  that $K$  is directly related to the amount of vegetation along ponds borders, since these areas with vegetation provide the habitat available for the larvae, where they have shade, food and {protection from the predators and from being carried away} (see {\em.e.g.} \cite{danis1997aquatic,Reis2015malj,kumar2006larvicidal} and references therein).
    
    Similar modeling choice for the recruiting term have been found in \cite{Moulay2011} and in \cite{DumontChiroleuDomerg2008,DumontTchuenche2012,DumontThuilliez2016}, the latter being our departure points for the general mosquito dynamic and breeding site control modeling.
    Additional details on {assumptions regarding} $K$ are given in the next subsection, where {border} vegetation appears as a time-variable function and $K$ is set to depend on the vegetation level.

    % The model assumes that aquatic mosquitoes are born at a rate $\alpha(M_S + M_I),$ while they die at a density-dependent rate $\alpha (M_S + M_I)L/K$ and density-independent rate $\mu_L.$ 
    % Density dependent mortality of {\em Anopheles} immatures due to crowding was found in experiments \cite{muriu2013larval}.  
    
\begin{table}[ht]
\centering
\begin{tabular}{cc}
    \hline
    \textbf{Parameter} & \textbf{Biological Meaning} \\
    \hline \hline 
    %\vspace{5pt} 
    $\alpha$ & {Maximum {\em per-capita} recruitment} rate {of {\em Anopheles} immatures}\\ 
    $K^{-1}$ & {Slope of the larval-density dependent recruitment rate}\\
    $K_p^{-1}$ & {$K^{-1}$ in fish ponds with predatory fishes}  \\
    $K_w^{-1}$ & {$K^{-1}$ in fish ponds without predatory fishes}  \\
    $\nu$ & {Immature}-to-adult transition rate\\
    $\mu_L$ &   {Immature mosquito natural death rate}\\   
    $\mu_p$ &  {Immature mosquito death rate}  due to predatory fish \\
    $\mu_M$ & Adult mosquito mortality rate\\
    $\beta_{hv}$ & Human-to-mosquito {parasite transmission} rate\\  
    $\beta_{vh}$ & Mosquito-to-human {parasite transmission} rate\\
    $\kappa$  & Recovery rate {of} infected {humans}\\
       $\tau$ & Period{icity} of {fish pond} vegetation cleaning, in days\\
    $r$ & Rate of growth of vegetation {along the fish pond borders} \\
    $\gamma$ & Proportion of vegetation removed as a function of time \\
    \hline
\end{tabular}   
\caption{Biological meaning of the parameters for systems \eqref{Ch2:Sec2:sys1} and \eqref{sec2:sys3}.}
\label{model:table1}
\end{table}

\subsection{Dynamics of Vegetation}\label{modeling:the_model:vegetation}
    Next we present a model for the growth of the border vegetation {and the cleaning practices of the fish farmers. The vegetation is mostly fast growing herbaceous plants that demand frequent clearing. We let $V$ be the amount of border vegetation at time $t,$ defined as a proportion of the total pond border. } \cite{Reis2015malj,Reis2015plos} 
    %Figure \ref{ch2:fig2} shows some examples of ponds, with low and with high amounts of border vegetation.
    
We assume that the percentage of border vegetation $V: \mathbb{R}_+\mapsto [0,1]$ grows proportionally to available space $1 - V$, and that {\em some} cleaning occurs every $\tau$ days, in which a part $\gamma$ of border vegetation is removed. Hence $\gamma:\mathbb{R}_+ \mapsto[0,1] $ is the fraction of vegetation cleaned as function of time. More precisely, at day $n\tau$, the proportion of vegetation removed is $\gamma(n\tau)$. \sol{Following ideas from Dumont {\em et al.} \cite{DumontTchuenche2012,DumontThuilliez2016}, we} propose the following system of impulsive ordinary differential equations to model the described situation:
\begin{subequations}
\label{Sec2:Sys2}
\begin{align}\frac{d V (t)}{d t}  &=  r (1 - V(t)),\quad \text{for } t \neq \tau n,  \\
    V (n \tau^+) & = V(n\tau^-)- \gamma (n \tau) V(n \tau^-),\quad \text{for }n \in \mathbb{N}\backslash\{0\},\\
        V(0) &= V_0, 
\end{align}
\end{subequations}
    where $r > 0$ is the vegetation growth rate (see Table \ref{model:table1}). 
    
    \sol{System \eqref{Sec2:Sys2} above can be seen as a particular {\em growth-disturbance model} for vegetation. For alternative choices of the cleaning actions, one can obtain different behaviours of the system, as shown in Reluga \cite{Reluga2004} and references therein. Here we impose $\gamma$ to asymptotically approach a constant value. Further generalizations are out of the scope of this article.}

\begin{proposition}\label{prop:veg_model_solutions}
    Given any initial condition $V_0 \in [0,1],$ there exists a unique solution of \eqref{Sec2:Sys2}, it is defined in $[0,+\infty)$, it is left-continuous and takes values in $[0,1]$.
\end{proposition}
  
\begin{proof}
\sol{This result can be deduced from more general established properties of impulsive differential equations, as found in \cite{bainov1993impulsive}. The proof in our context is quite straightforward. Nevertheless, for a the sake of completeness, we include it in this article.}

The proof follows by induction on $n$. 
For $n=1$, let us consider the interval $[0,\tau)$. It is straightforward that  $V(t) = 1 - (1 - V_0) {e^{-rt}} $
is a solution of \eqref{Sec2:Sys2} in $[0,\tau]$. Moreover, it is unique, it is continuous and takes values in $[0,1]$.
Suppose now that there exists a unique solution ${V_n}$ defined on $[0,n\tau]$ that is left-continuous and takes values in $[0,1]$. Set 
$V_{n}^+ \coloneqq \Big(1 - \gamma\big(n \tau\big)\Big) {V_n}(n\tau^-). $ Since ${V_n}(n\tau^-) \in [0,1]$ and $1 - \gamma\big(n \tau\big) \in [0,1]$, we have that $V_{n}^+ \in [0,1]$. Consider the initial value problem on $[n \tau, (n+1)\tau]:$
\[
\dot V(t) = r(1 - V(t)),\qquad V(n\tau^+) = V_{n}^+. \]
It follows easily that $V(t) = 1 - (1 - V_{n}^+) {e^{-rt}}$ is its unique solution. Set now
\begin{equation*}
    {V_{n+1}}(t) \coloneqq
\begin{cases}
    {V_{n}}(t),\quad \text{for }t \in [0,n\tau]\\
    1 - (1 - V_{n}^+){e^{-rt}},\quad\text{for }t\in  (n \tau, (k+1)\tau].
\end{cases}
\end{equation*}
The function ${V_{n+1}}$ is the unique solution of  \eqref{Sec2:Sys2} in $[0,(n+1)\tau]$ and satisfies the desired properties. This concludes the inductive step and the proof.
\end{proof}  

\subsection{The complete Model}\label{modeling:the_model:full}
 
 We conclude the construction of our model by putting together \eqref{Ch2:Sec2:sys1} and \eqref{Sec2:Sys2} with the following additional considerations. 
 {Let us consider that there are two types of fish ponds,} namely {\em with} and {\em without predatory fish}. This differentiation induces  a splitting in the aquatic stage $L$ into $L_p$ and $L_w$, corresponding to larvae in ponds with and without predatory fish, respectively.  
 {In the absence of predation, larvae can graze in a larger area of vegetation thus relaxing density-dependent effects on the recruitment rate, such as competition for resources. 
 This is modeled by defining $K_p^{-1}$ and $K_w^{-1}$ as the slope of the larval-density dependent recruitment rate in ponds with and without predatory fishes, respectively, and letting $K_p < K_w$.} 
 
{Besides the increased density effects, the presence of fishes also acts directly on the mortality of larvae by predation. 
Thus the population $L_p$ suffers an extra mortality rate $\mu_p (1 - V )$, where $\mu_p$ is the maximum predation rate which decreases as hiding places in the vegetation increases.  
adult mosquitoes that emerge from $L_p$ and $L_w$ merge into a single Adult compartment $M_S$.}
{To complete the cycle, female mosquitoes must choose where to lay eggs.} The biological literature shows evidence that {\em Anopheles gambiae} females, like other mosquito species, choose among breeding sites. 
 Presence of predators is avoided while the presence of conspecifics tend to be attractive as it may indicate suitable habitats \cite{munga2014}.
{We model this behavior by attributing a higher probability of oviposition in ponds with more vegetation and without fishes. 
This is achieved by  multiplying the recruitment rates of $L_p$ and $L_w$ by $K_p/(K_w +K_p)$ and $K_w/(K_w +K_p)$, respectively, noting that these fractions add up to one and do not affect the total recruitment rate.} {The terms} $K_p$ and $K_w$ are positive-valued strictly increasing continuous functions of border vegetation.
  
The following system describes the joint dynamics of malaria, aquatic-stage and adult mosquitoes, and vegetation:
\begin{equation}
\begin{split}\label{sec2:sys3}
    {\dot V}  &=  r (1 - V), \\
    V (n\tau^+) &= V(n\tau^-) - \gamma (n \tau) V(n\tau^-),\quad \text{for}\  n \in \mathbb{N},\\
     {\dot L_p} &= \alpha \frac{K_p(V)}{K_w(V) + K_p(V)} (M_I + M_S)\left(1 - \frac{L_p}{K_p(V)} \right) - (\nu + \mu_L + \mu_p(1 - V))L_p,\\
     %\solb{L_p(n\tau^+)} & \solb{= \min\left\{ L_p(n\tau^-) , K_p(V(n\tau^+))\right\},\quad \text{for}\  n \in \mathbb{N},}\\ 
     {\dot L_w} &= \alpha \frac{K_w(V)}{K_w(V) + K_p(V)} (M_I + M_S) \left(1 - \frac{L_w}{K_w(V)} \right) - (\nu + \mu_L)L_w,\\
     %\solb{L_w(n\tau^+)} & \solb{= \min\left\{ L_w(n\tau^-) , K_w(V(n\tau^+))\right\},\quad \text{for}\  n \in \mathbb{N},}\\ 
     {\dot M_S} &= \nu (L_w + L_p) -{\beta_{hv} I M_S} - \mu_M M_S,\\
     {\dot M_I} &= {\beta_{hv} I M_S} - \mu_M M_I,\\
     {\dot S} &= -\beta_{vh} S M_I + \kappa I,\\
     {\dot I} &= \beta_{vh} S M_I - \kappa I.
\end{split}
\end{equation}
where the biological meaning of the parameters is given in Table \ref{model:table1}. The compartmental diagram for system \eqref{sec2:sys3} is shown in Figure \ref{modeling:diagrama_full}. 
\begin{figure}[ht]
\centering
\input{diagrama_s3.tex}\\
\caption[Compartmental diagram for system \eqref{sec2:sys3}.]{Compartmental diagram for the dynamics described in system \eqref{sec2:sys3}. \fel{Solid arrows represent the flow of individuals from one compartment to the other. Dashed arrows stand for some form of influence that a  compartment has in the flow entering of exiting another compartment. We omit some arrows and the parameters in order to keep the diagram compact and legible.}}
\label{modeling:diagrama_full}

\end{figure}
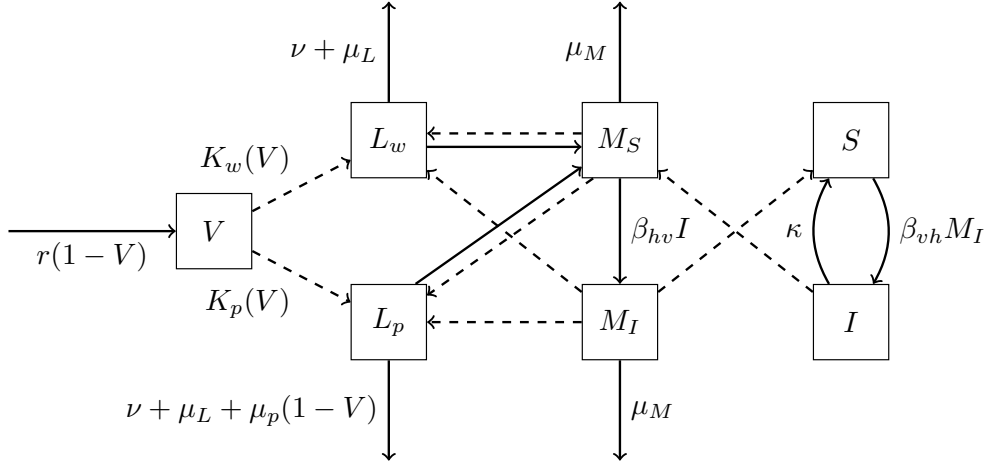

  We now state and prove results concerning properties of solutions to initial value problems for system \eqref{sec2:sys3}.  Consider the compact set
\begin{equation}\label{sys3compact}
    \mathcal{K} \coloneqq 
    \left\{
    \begin{aligned}
        & (V,M_S,M_I,L_p,L_w,S,I) \in \mathbb{R}^7_+ : V \leq 1,\, S + I \leq 1, \\
        & M_S + M_I \leq (\nu/\mu_M)(L_p + L_w),\, L_p  \leq \max_{V \in [0,1]}K_p(V),\, L_w  \leq \max_{V \in [0,1]}K_w(V)
    \end{aligned}
    \right\}.
\end{equation}

\begin{proposition}\label{Sys3WellPosedProp}
Given an initial condition $x_0$ belonging to $\mathcal{K}$, there exists a unique solution of \eqref{sec2:sys3} with value $x_0$ at $t=0$, it is defined in $[0,+\infty)$ and  remains in $\mathcal{K}.$ The last assertion states that $\mathcal{K}$ is positively invariant under \eqref{sec2:sys3}. 
\end{proposition}    
    
\begin{proof}
The equation for the vegetation can be solved independently of the rest of the system, and gives a solution $V$ that is left-continuous and bounded (thanks to Proposition \ref{prop:veg_model_solutions}). Inserting this solution in the equations for $M_S,M_I,L_p,L_w,S,I$ yields a system that satisfies the conditions of Theorem \ref{sec3:theo1} in the Appendix \ref{Ch3:sec2}, which guarantees local existence.

   Proof of positive invariance of $\mathcal{K}$ follows easily by checking that, when a trajectory issuing from $\mathcal{K}$ hits the boundary of $\mathcal{K},$ the dynamics pushes the trajectory inwards the set.
    
    The existence of a global solution is a consequence of Theorems \ref{sec3:theo1}-\ref{sec3:theo3} in Appendix \ref{Ch3:sec2}.
\end{proof}

\section{Asymptotic behaviour}\label{modeling:analysis}

   In this section we study the asymptotic behaviour of system \eqref{sec2:sys3}. We split this analysis in three parts. In Subsection \ref{Asy_Beh_Veg_SSct} our focus is the vegetation system \eqref{Sec2:Sys2}. We prove, in Proposition \ref{ch2:sec3:prop1}, that if cleaning at times $t=n\tau$ tends to a constant value then border vegetation approaches a $\tau$-periodic solution. Subsection \ref{Mos_Dyn_SSct} is dedicated to the study of the behaviour of mosquito population when border vegetation is periodic, and \ref{Dis_Beh_SSct} investigates malaria incidence under periodic vegetation cleaning. \sol{By applying the results on cooperative systems stated in Appendix \ref{appendix_cooperative},}  we derive sufficient conditions for the trajectories of \eqref{sec2:sys3} to asymptotically approach  either a positive solution or the origin. 
   
   \sol{At this point, it is worth commenting that the main interest of our research is finding conditions under which malaria is eradicated and these are established in Theorem \ref{Theo 3.3}. Additionally, Theorem \ref{Theo 3.2}  states conditions for mosquito population elimination. 
   The conditions in Theorem \ref{Theo 3.2} turn out to require quite more complex expressions than the ones in Theorem \ref{Theo 3.3}, but they are relevant from a theoretical point of view and the expressions can be simplified in particular cases.} 

\subsection{Asymptotic behaviour of vegetation}\label{Asy_Beh_Veg_SSct}

We focus now on the vegetation model \eqref{Sec2:Sys2}. We aim to prove that if cleaning is done every $\tau$ days and approaches a constant value, then the vegetation converges to a periodic solution that is discontinuous at the cleaning times. We get the following result.

\begin{proposition}[Periodic solutions for vegetation]\label{ch2:sec3:prop1}\\
Let $\gamma^* \in [0,1]$,
and assume that $\gamma\colon [0,+\infty) \to [0,1]$ is such that $\displaystyle\lim_{t \to \infty} \gamma(t) = \gamma^*$.
Then any trajectory $V$ of \eqref{Sec2:Sys2}  with initial condition in $[0,1]$ converges to the
 periodic solution
\begin{equation}\label{ch2:sec3:Vper_sol}
\begin{split}
    V_{\rm per}(t) &\coloneqq 1 - \frac{\gamma^* e^{-r(t - n\tau)}}{1 - (1 - \gamma^* )e^{-r\tau}},\quad  \text{for}\ t \in [n\tau, (n+1)\tau),\\
\end{split}
\end{equation}
   in the following sense:
\begin{equation}
    \begin{split}
        \max_{t \in [n\tau,(n + 1)\tau)}  |V(t) - V_{\rm per}(t)| \to 0,\quad \text{as}\ n \to \infty. \\
    \end{split}
\end{equation} 
\end{proposition}

Proposition \ref{ch2:sec3:prop1} follows from  a direct application of Lemma \ref{Ch2:sec3:lem1} in the Appendix \ref{appendix_impulsive}
 by setting $\gamma_n \coloneqq \gamma(n \tau)$. We conclude that, under periodic cleaning, the border vegetation approaches a periodic behaviour. Figure \ref{modeling:Veg_multiple_x0} illustrates Proposition \ref{ch2:sec3:prop1}, by showing the solutions to the impulsive initial value problem \eqref{Sec2:Sys2} for a range of initial values in $[0,1]$ and for values of \fel{$\gamma_n = (0.65 + (1 - 0.65)\frac{5}{n\tau})$,} so that the values of $\gamma_n$ converge to $0.65$. Note that at each point of discontinuity, the trajectories come closer together, and asymptotically approach  $V_{\rm per}$.

\begin{figure}[ht]
\begin{center}
\indent{\centering\includegraphics[width=0.8\linewidth]{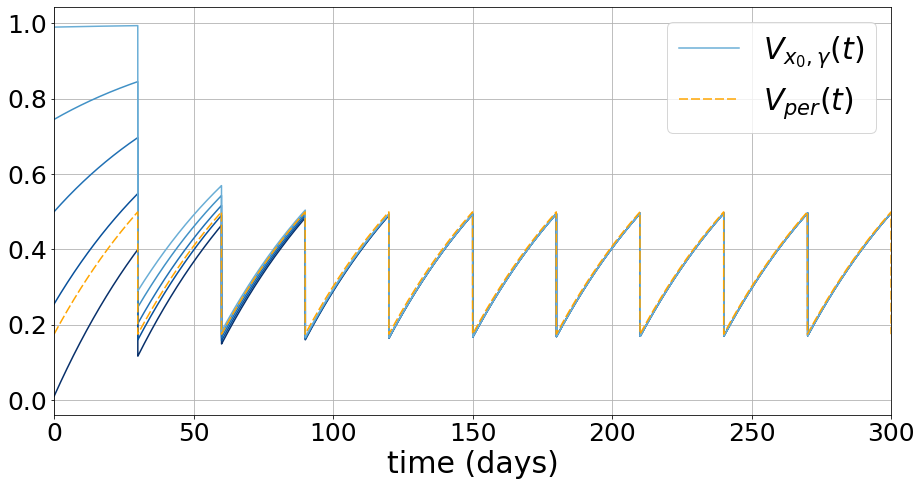}}\\
  \indent 
\caption[Impulsive behaviour of {system} \eqref{Sec2:Sys2}.]{Impulsive behaviour of {system} \eqref{Sec2:Sys2}, for different initial values. $V_{\rm per}$ is shown in orange.}\label{modeling:Veg_multiple_x0}
\end{center}
\end{figure}

%\begin{corollary}\label{ch2:sec3:cor2}
%Consider the equations describing the behaviour of vegetation
%\begin{equation}
%\begin{split}
%    {\dot V}  &=  r_n (1 - V(t)) \\
%    \Delta V &(n \tau) = - \gamma_n (H(n \tau)) V(n \tau),\ n \in \mathbb{N}.\\
%\end{split}
%\end{equation}
%Assuming $H(t)$ converges to $H^*$. It follows from \ref{Ch2:sec3:Cor1} that there exists periodic solutions $V_{\rm per}$ and any solution $V$ converge globally asymptotically to $V_{\rm per}$.
%\end{corollary}

%From \ref{Ch2:sec3:cor2}, proposition \ref{ch2:sec3:prop1} follows.

\subsection{Asymptotic behaviour of mosquito population}\label{Mos_Dyn_SSct}
We now turn to the analysis of the asymptotic behaviour of the mosquito population, {\em assuming that border vegetation {clearing} is periodic} and equal to $V_{\rm per}$. Our aim is to derive sufficient conditions on the parameters of the system for the mosquito population to either converge to $0$, or have asymptotic periodic behaviour with positive values.

Recalling system \eqref{sec2:sys3},
there are two aspects to consider. First, that {the periodic} vegetation {clearing} behaviour impacts the system through  {the terms}
\[
\fel{\tilde K_w(t) \coloneqq K_w(V_{\rm per}(t)),\quad 
\tilde K_p(t) \coloneqq K_p(V_{\rm per}(t))}
\]
which are also periodic.  Second, by setting $M \coloneqq M_S + M_I,$ we can isolate mosquito population dynamics in the following subsystem:
\begin{equation}
\label{system 3.2}
\begin{split}
     {\dot L_p} &= \alpha \frac{{\tilde K}_p}{{\tilde K}_w + {\tilde K}_p} M\left(1 - \frac{L_p}{{\tilde K}_p} \right) - (\nu + \mu_L + \mu_p(1 - V_{\rm per}))L_p,\\
     {\dot L_w} &= \alpha \frac{{\tilde K}_w}{{\tilde K}_w + {\tilde K}_p} M \left(1 - \frac{L_w}{{\tilde K}_w} \right) - (\nu + \mu_L)L_w,\\
     {\dot M} &= \nu (L_p + L_w) - \mu_M M.
\end{split}
\end{equation}
      Let us write the system in the compact form ${\dot X}(t) = F(t,X(t))$ for $X:=(L_p,L_w,M)^\top$, $F$ being the function on the r.h.s. of \eqref{system 3.2}, and let us consider  the set 
\begin{multline}\label{Kaex}
    \mathcal{K}_X \coloneqq \big\{(t,L_p,L_w,M) \in \mathbb{R}^4_+ : \\
    \sol{M \leq \frac{\nu}{\mu_M}\Big(\max_{[0,\tau)}{\tilde K}_p(t) + \max_{[0,\tau)}{\tilde K}_w(t)\Big),\, L_p  \leq \tilde K_p(t),\, L_w  \leq \tilde K_w(t)
    \big\}},
\end{multline}
     \fel{on which $F$ is cooperative.} Due to the discontinuities and periodicity of $V_{\rm per}$, the function $F$ itself is piecewise-continuous  and $\tau$-periodic both w.r.t. the time variable $t$.
% New paragraph     
 \fel{  The difficulty in analyzing system \eqref{system 3.2} arises when the discontinuities of $V_{\rm per}$ at the times $n\tau$ may steer the state of the system out of the set $\mathcal{K}_X$ where it is cooperative.}
%end new paragraph
Before stating the main result of this section, let us recall the formula for the {\em basic offspring number}
$$    
    \mathcal{N} \coloneqq \frac{ \alpha \nu}{\mu_M (\nu + \mu_L)},  
$$    
    that represents the {\em average number of offspring that an individual produces during its lifespan} \cite{yang2009assessing, ferreira2014ecological}, and is considered a measure of the growth of a population. Additionally, \fel{let 
\begin{equation}
    \label{Np}
    \mathcal{N}_p(t) \coloneqq \frac{\nu \alpha}{\mu_M\big(\nu + \mu_L + \mu_p(1 - V_{\rm per}(t))\big)}
\end{equation} }
\sol{be a ``modified'' basic offspring number.}

\begin{theorem}[Asymptotic behaviour of the mosquito population]\label{Theo 3.2}
    The following assertions \sol{on the trajectories of system \eqref{system 3.2}} hold.

\fel{\noindent \em \textit{(i) If }}
\begin{equation}\label{conditions_for_0}
\begin{aligned}
\sol{ \mathcal{N} \max_{t\in [0,\tau)} {\frac{{\tilde K}_w}{ {\tilde K}_w + {\tilde K}_p}} + \mathcal{N}_p(\tau^-) \max_{t\in [0,\tau)}{\frac{{\tilde K}_p}{ {\tilde K}_w + {\tilde K}_p}} \leq 1}
\end{aligned}
\end{equation}
\fel{\textit{is satisfied, then the trajectories of \eqref{system 3.2} asymptotically approach the origin for any initial condition in $\mathcal{K}_X$.}} 

\fel{\noindent  (ii)  If}
\begin{equation}\label{conditions_for_per}
\sol{
 \frac{\min_{[0,\tau)}{\left( \frac{ {\tilde K}_w (t)}{ {\tilde K}_w(t) + {\tilde K}_p (t)} \right)}}{\mathcal{N}^{-1} + \frac{ {\tilde K}_w (\tau^-) + {\tilde K}_p(\tau^-) }{{\tilde K}_w(0) + {\tilde K}_p(0)}} + 
 \frac{\min_{[0,\tau)} \frac{ {\tilde K}_p (t)}{ {\tilde K}_w(t) + {\tilde K}_p (t)}}{\mathcal{N}_p(0)^{-1} + \frac{ {\tilde K}_w (\tau^-) + {\tilde K}_p(\tau^-) }{{\tilde K}_w(0) + {\tilde K}_p(0)}} > 1
}
\end{equation}
\fel{is satisfied, then for all initial conditions in $\mathcal{K}_X\backslash \{0\}$, solutions of \eqref{system 3.2} are bounded above and below by asymptotically periodic solutions.}
\end{theorem}

\begin{proof}
\sol{This proof is based on the application of Theorem \ref{Ch3:Main_Theo} in the Appendix (see also Remark \ref{RemThms}), which follows from results on cooperative system with nonlinear concavities proved by Jiang \cite{jifa1993cooperative} and Smith \cite{smith1986cooperative}}.

\fel{As mentioned above, the difficulty in dealing with system \eqref{system 3.2} is that the jumps in ${\tilde K}_p,{\tilde K}_w$ may stir the trajectories outside the set where the equation is cooperative. We can write  \eqref{system 3.2} as 
\[ \dot x(t) = F(t,x) =  A(t) x(t) - f(t,x),\]
where $x(t) = (L_p(t), L_w(t), M(t))^T$,
\[A(t) \coloneqq \begin{pmatrix}
    - (\nu + \mu_L + \mu_p(1 - V_{\rm per}))  & 0 &  \alpha \frac{{\tilde K}_p}{{\tilde K}_w + {\tilde K}_p} \\
    0 & -(\nu + \mu_L) & \alpha \frac{{\tilde K}_w}{{\tilde K}_w + {\tilde K}_p} \\
    \nu & \nu & -\mu_M
\end{pmatrix},\]
and $f(t,x) \coloneqq  \frac{\alpha}{{\tilde K}_w + {\tilde K}_p} (L_p M, L_w  M, 0)^T$, which is non-negative over $\mathcal{K}_X$. Note that $A(t)$ has non-negative off diagonal terms.
Consider the system 
\[ \dot z(t) = A(t)z(t).\]
We have that latter system is cooperative  in $\mathcal{K}_X$, and that 
\[
A(t)x \geq A(t)x - f(t,x),\quad \text{for all }(t,x) \in \mathcal{K}_X.
\] 
Moreover, both systems (for $x$ and $z$) satisfy the uniqueness property for initial value problems. Therefore we can apply Kamke's Theorem \cite{smith1986cooperative,jifa1993cooperative} to conclude that, for the same initial condition $z(0) = x(0) = x_0$, the corresponding trajectories satisfy $z(t) \geq x(t)$.}

\sol{At this point, let us recall the notation $\overline g,$ $\underline g$, $\overline A$ and $\underline A$ introduced in the notation Section \ref{SecNotation}.}
\fel{Now, from \cite[Proposition 5.1]{jifa1993cooperative}, we know that the Floquet multiplier of maximal modulus $\lambda$ of $A(t)$ satisfies $\exp(\tau s(\underline{A} )) \leq \lambda \leq \exp(\tau s({\overline A}))$, whereas from \cite[Theorem 5.1]{jifa1993cooperative} we know that for a cooperative matrix $M$, $s(M) > 0$ if, and only if, there exists a negative principal minor of $-M$. We conclude that, if all principal minors of $-{\overline A}$ are non-negative, then $s(\overline A) < 0$ and $\lim_{t\to\infty} z(t) = 0$, which implies that $\lim_{t\to\infty} x(t) = 0$. 
For the principal minors we get the following expressions:}
\begin{equation*}
%\resizebox{\textwidth}{!}{
\fel{
\begin{aligned}
& \bar p_1 = \left( \nu + \mu_L + \mu_p\Big(1-V(\tau^-)\Big)\right) \left(\nu + \mu_L \right)
\mu_M \\
& \quad - \alpha \nu \left(\nu + \mu_L  \right) \overline{\frac{ {\tilde K}_p}{{\tilde K}_p + {\tilde K}_w}}
 - \alpha \nu  \left(\nu + \mu_L + \mu_p\Big(1-V_{\rm per}(\tau^-)\Big) \right)  \overline{\frac{ {\tilde K}_w}{{\tilde K}_p + {\tilde K}_w}}
,\\
& \bar p_2 = \mu_M \left(\nu + \mu_L  \right) - \alpha \nu  \overline{\frac{ {\tilde K}_w}{{\tilde K}_p + {\tilde K}_w}},\\
& \bar p_3 = \mu_M 
\left( \nu + \mu_L + \mu_p\Big(1-V_{\rm per}(\tau^-)\Big)\right) 
- \alpha \nu  \overline{\frac{ {\tilde K}_p}{{\tilde K}_p + {\tilde K}_w}}, \\
& \bar p_4 = 
\left( \nu + \mu_L + \mu_p\Big(1-V_{\rm per}(\tau^-)\Big)\right)
\left(\nu + \mu_L \right),\\
&\bar p_5 = \left( \nu + \mu_L + \mu_p\Big(1-V_{\rm per}(\tau^-)\Big)\right),
\ \bar p_6 = \left(\nu + \mu_L \right), 
\ \bar p_7 = \mu_M.
\end{aligned}
}
\end{equation*}
\fel{Principal minors $\bar p_4, \bar p_5, \bar p_6$ and $\bar p_7$ are trivially non-negative. We want conditions that ensure that $p_1 \geq 0$, $p_2 \geq 0$ and $p_3 \geq 0$. We find that}
\fel{\begin{equation*}
\begin{split}
    \bar p_1 &= -\overline A_{1,1}{\bar p_2} + {\overline A}_{2,2} {\overline A}_{1,3} {\overline A }_{3,1} =  -\overline A_{2,2}{\bar p_3} + {\overline A}_{1,1} {\overline A}_{2,3} {\overline A }_{3,2} \\ 
    &= -\overline A_{1,1}{\bar p_2} -\overline A_{2,2}{\bar p_3} + {\overline A }_{1,1}{\overline A }_{2,2}{\overline A }_{3,3}
\end{split}
\end{equation*}}
\fel{From latter equation, we have that $\bar p_1 \leq -\overline A_{1,1}\bar p_2$, from which follows that $\bar p_1 \geq 0 \Rightarrow \bar p_2 \geq 0$. Analogously, $\bar p_1 \geq 0 \Rightarrow \bar p_3 \geq 0$. Therefore we only need to impose $\bar p_1 \geq 0$. This gives}
\begin{equation*}
\begin{gathered}
\fel{\frac{\bar p_2}{{\overline A}_{2,2}\ {\overline A}_{3,3}} + \frac{\bar p_3}{{\overline A}_{1,1}\ {\overline A}_{3,3}} \geq 1.}
\end{gathered}
\end{equation*}
\fel{If we define $\mathcal{N}_p$ as in \eqref{Np},
then we can write latter condition as} \eqref{conditions_for_0}.

\fel{Now, consider the matrix
\[
B \coloneqq \frac{\alpha \nu}{\mu_M}\frac{{\tilde K}_w (\tau^-) + {\tilde K}_p(\tau^-)}{{\tilde K}_w(0) + {\tilde K}_p(0)}   \begin{pmatrix}
    1  & 0 & 0 \\
    0 & 1 & 0 \\
     0 & 0 & 0\\
\end{pmatrix}.
\]
Note that $f(t,x) \leq Bx$, so $A(t)x - Bx \leq A(t)x - f(t,x)$. If we define $w$ such that $\dot w = (A(t) - B)w$, $w(0) = x(0)$, then $w$ satisfies a cooperative differential equation, and by Kamke's Theorem $w(t) \leq x(t)$. We can apply the same reasoning as before to conclude that if at least one principal minor of $-\big( \underline{ A - B} \big)$ is negative, \sol{then the trajectories of the system in $w$ approach a strictly positive periodic solution which bounds $x$ from below}. The principal minors of $-\big(\underline{ A - B}\big)$ are: }
\begin{equation*}
\resizebox{\textwidth}{!}{\fel{$
\begin{aligned}
& \underline{p_1} = \left( \nu + \mu_L + \mu_p\Big(1-V_{\rm per}(0)\Big) + \frac{\alpha \nu}{\mu_M}\right) \left(\nu + \mu_L + \frac{\alpha \nu}{\mu_M}\right)
\mu_M \\
& \quad - \alpha \nu \left(\nu + \mu_L  + \frac{\alpha \nu}{\mu_M}\right) \underline{\frac{ {\tilde K}_p}{{\tilde K}_p + {\tilde K}_w}}
 - \alpha \nu  \left(\nu + \mu_L + \mu_p\Big(1-V_{\rm per}(0)\Big)  + \frac{\alpha \nu}{\mu_M} \right) \underline{ \frac{{\tilde K}_w}{{\tilde K}_p + {\tilde K}_w}}
,\\
& \underline{p_2} = \mu_M \left(\nu + \mu_L + \frac{\alpha \nu}{\mu_M} \right) - \alpha \nu \underline{ \frac{{\tilde K}_w}{{\tilde K}_p + {\tilde K}_w}},\\
& \underline{p_3} = \mu_M 
\left( \nu + \mu_L + \mu_p\Big(1-V_{\rm per}(0)\Big) + \frac{\alpha \nu}{\mu_M}\right) 
- \alpha \nu \underline{ \frac{{\tilde K}_p}{{\tilde K}_p + {\tilde K}_w}}, \\
& \underline{p_4} = 
\left( \nu + \mu_L + \mu_p\Big(1-V_{\rm per}(0)\Big) + \frac{\alpha \nu}{\mu_M}\right)
\left(\nu + \mu_L + \frac{\alpha \nu}{\mu_M} \right),\\
& \underline{p_5} = \left( \nu + \mu_L + \mu_p\Big(1-V_{\rm per}(0)\Big) + \frac{\alpha \nu}{\mu_M}\right),
\quad \underline{p_6} = \left(\nu + \mu_L + \frac{\alpha \nu}{\mu_M} \right), 
\ \underline{p_7} = \mu_M.
\end{aligned}$}}
\end{equation*}
\fel{Denoting terms of the matrix $\underline{A - B}$ by $\underline{c}_{i,j}$, we have that
\[\underline{p}_1 = -\underline{c}_{1,1} \underline{p}_2 + \underline{c}_{2,2} \underline{c}_{1,3}\underline{c}_{3,1} = -\underline{c}_{2,2} \underline{p}_3 + \underline{c}_{1,1}\underline{c}_{2,3}\underline{c}_{3,2}
= -\underline{c}_{1,1} \underline{p}_2 - \underline{c}_{2,2} \underline{p}_3 + \underline{c}_{1,1}\underline{c}_{2,2}\underline{c}_{3,3}\]
From the latter equation, we have that $\underline{p}_1 \leq \underline{p}_2$ and $\underline{p}_1 \leq \underline{p}_3$, so ``at least one negative principal minor" is equivalent to $\underline{p}_1 < 0$, from which we have}
\begin{equation*}
\begin{gathered}
\fel{\frac{\underline{p}_2}{\underline{c}_{2,2}\ \underline{c}_{3,3}} + \frac{\underline{p}_3}{\underline{c}_{1,1}\ \underline{c}_{3,3}} < 1.}
\end{gathered}
\end{equation*}
\fel{Writing latter condition in full} \sol{leads to} \eqref{conditions_for_per}.
\fel{Note that the conditions for $w$ to be asymptotically periodic imply that $z$ is also asymptotically periodic, so in this case $x$ is bounded above and below by asymptotically periodic functions, as stated in the theorem.}
\end{proof}
\begin{remark}
    \fel{The numerical simulations in Section \ref{Ch2:NumSimSection} show that system \eqref{system 3.2} has, in general, asymptotic positive periodic behaviour.}
\end{remark}
\if{
First, some comments about Theorem \ref{Theo 3.2}. All the conditions depend on the basic offspring number $\mathcal{N}$, and they cannot be simplified further, as the maxima (or minima) of $\displaystyle\frac{K_w}{K_p + K_w}(V_{\rm per})$ and of $\displaystyle\frac{K_p}{K_p + K_w}(V_{\rm per})$ may not occur at the same time $t$. If they actually occur at the same time $t$, then the expressions in \eqref{Theo 3.2 C1} reduces to $1$.

\eqref{Theo 3.2 C1} and \eqref{Theo_3.2_Cond_2} do not exhaust the possibilities for system, as $\mathcal{N}^{-1}$ can be larger than the sum of minima in \eqref{Theo_3.2_Cond_2} and smaller than the sum of maxima in \eqref{Theo 3.2 C1}. Numerical simulations in the current section illustrate that in the latter case, the mosquito population assume periodic behaviour close to $0$.
    
For the {\em periodic mosquito population} $M_{\rm per}$, we do not have an analytical expression as we have for $V_{\rm per}$ in Proposition \ref{ch2:sec3:prop1}. We are specially concerned with the effects of a reduction in the period $\tau$, which means an increase in cleaning frequency. This reduction in period alters the value of the right-hand side in both \eqref{Theo 3.2 C1} and \eqref{Theo_3.2_Cond_2}, by restricting the values of \textit{border vegetation} $V_{\rm per}$ to lower values (as the period $\tau$ reduces, vegetation has less time to grow between cleaning episodes). Even if the reduction in period does not change the behaviour from asymptotical periodicity to asymptotical extinction, it does reduce the average size of the mosquito population, which has attenuating effects on the incidence of the disease.

We now prove the following corollary to Theorem \ref{Theo 3.2}.
%\if{
\begin{corollary}\label{analysis:corollary}
Assume that $K_w$ and $K_p$ are proportional, that is
$$
K_p/K_w \equiv \rho,
$$
with $\rho$ constant.
Then the conditions for Theorem \ref{Theo 3.2} are simplified to
\begin{enumerate}
    \item If $\mathcal{N} \leq 1$, then $\lim_{t \to +\infty} X(t) = 0$ for all feasible initial conditions.
    \item If $\mathcal{N} > 1 + \rho$, then there is a strictly positive periodic solution which attracts all initial conditions.
\end{enumerate}
\end{corollary}

\begin{proof}

From the proportionality of $K_w(V)$ and $K_p(V)$, we have
$$
    \frac{K_w}{K_p + K_w} (V) = \frac{1}{1 + \rho}, \quad
    \frac{K_p}{K_p + K_w} (V) = \frac{\rho}{1 + \rho}.
$$
As $\rho$ is constant, the expression in \eqref{Theo 3.2 C1} becomes
$$
\mathcal{N}^{-1} \geq \frac{1}{1 + \rho} + \frac{\rho}{1 + \rho} = 1,
$$
and the expression for \eqref{Theo_3.2_Cond_2} becomes
$$
1 + \rho < \mathcal{N},
$$
from which the statement of the corollary follows.
\end{proof}
}\fi

\subsection{Asymptotic behaviour of malaria}\label{Dis_Beh_SSct}

Following the discussion of the previous subsection, we now analyze the behaviour of the {\em infected} human and mosquito populations in the same manner. Assuming we are in situation (ii) of Theorem \ref{Theo 3.2}, let $M_{\rm per}$ be the periodic solution for the mosquito population. We get that the infected components $(M_I,I)$ follow the dynamics:
\begin{equation}\label{Ch2:Sec3:ssec3:eq1}
\begin{split}
    \dot M_I &= \beta_{hv}I (M_{\rm per} - M_I) - \mu_M M_I;\\
    \dot I &= \beta_{vh} (1 - I) M_I - \kappa I.
\end{split}
\end{equation}
Note that system \eqref{Ch2:Sec3:ssec3:eq1} is cooperative in the time-dependent domain 
\[
\mathcal{K}_Y = \{(t,M_I,I) \in \mathbb{R}^3_+ : M_I \leq M_{\rm per}(t),\, I \leq 1 \}.
\]
We can apply Theorem \ref{Ch3:Main_Theo} to arrive at the following result:

\begin{theorem}[Limit Behaviour of Disease]\label{Theo 3.3}
    
\sol{(i)}    If 
%\[\fel{\frac{\beta_{vh}\beta_{hv}\max_{t \in [0,\tau)}M_{\rm per}(t)}{\kappa \mu_M }\leq 1,}\]
\begin{equation}\label{condi}
{
\max_{t \in [0,\tau)} \sqrt{\frac{\beta_{vh}\beta_{hv}M_{\rm per}(t)}{\kappa \mu_M }}\leq 1,}
\end{equation}
then the trajectories $(M_I,I)$ of \eqref{Ch2:Sec3:ssec3:eq1} verify $\displaystyle\lim_{t \to +\infty} (M_I,I)(t) = 0$ for all initial conditions in $\mathcal{K}_Y$.
    
\sol{(ii)} If 
%\[\fel{\frac{\beta_{vh}\beta_{hv}\min_{t \in [0,\tau)}M_{\rm per}(t)}{\kappa \mu_M } > 1,}\]
\begin{equation}\label{condii}
    {\min_{t \in [0,\tau)} \sqrt{\frac{\beta_{vh}\beta_{hv} M_{\rm per}(t)}{\kappa \mu_M }} > 1,}
\end{equation}
then there exists a strictly positive periodic solution of \eqref{Ch2:Sec3:ssec3:eq1} which attracts all initial conditions in $\mathcal{K}_Y \backslash \{0\}.$
\end{theorem}

\begin{proof}
\sol{As done for previous theorem, this proof is based on the application of Theorem \ref{Ch3:Main_Theo} in the Appendix (see also Remark \ref{RemThms}).}

\fel{In contrast to Theorem  \ref{Theo 3.2}, there is no need to  \sol{split the analysis in two cases} for the trajectories of \eqref{Ch2:Sec3:ssec3:eq1} never leave the set $\mathcal{K}_Y$ for initial conditions  in $[0,M_{\rm per}(0)]\times [0,1]$ \sol{and then Theorem \ref{Ch3:Main_Theo} is applicable for all trajectories starting in $\mathcal{K}_Y$}. We can calculate A to be}
\begin{equation}
  \fel{A = 
  \begin{pmatrix}
       -\mu_M & \beta_{hv}M_{\rm per} \\
         \beta_{vh}  &  -\kappa
\end{pmatrix}
}
\end{equation}
The principal minors of $-\overline A$ are
\fel{
\begin{equation}
\begin{split}
    {\bar p}_1 &= \mu_M \kappa - \beta_{hv} \beta_{vh} {\overline{M}_{\rm per}};\\ 
  {\bar p}_2 &= \mu_M;\\
  {\bar p}_3 &= \kappa.
 \end{split}
\end{equation}
}
Both ${\bar p}_2$ and ${\bar p}_3$ are always non-negative. We will therefore focus on the condition that ${\bar p}_1 \geq 0$, from which we get
\[
     \frac{\beta_{hv} \beta_{vh} {\overline{M}_{\rm per}}}{\kappa \mu_M} \leq 1.
\]
If latter inequality is satisfied, then trajectories of \eqref{Ch2:Sec3:ssec3:eq1} converge to $0$ for all initial conditions{, as stated in item (i), noting that the square root does not change the inequality. A biological interpretation of the quantities that appear in \eqref{condi}-\eqref{condii} is given after this proof.}

On the other hand, if we consider the principal minors of $-\underline A$, we arrive at
\[
 \frac{\beta_{vh}\beta_{hv}\underline{M}_{\rm per}}{\kappa\mu_M} > 1.
\]
If latter inequality is satisfied, then there exists a strictly positive periodic solution of \eqref{Ch2:Sec3:ssec3:eq1} which attracts all feasible initial conditions, as stated {in item (ii)}.
\end{proof}

It is interesting to compare the {formulae \eqref{condi} and \eqref{condii} of the sufficient conditions} of Theorem \ref{Theo 3.3} to the {\em basic reproduction number} $\mathcal{R}_0$ for the disease \sol{(see \cite{lopez2002threshold})}, defined for a constant mosquito population $M$, which is given by
\[
{\mathcal{R}_0 \coloneqq \sqrt{\frac{\beta_{hv} \beta_{vh} M}{\kappa \mu_M}}
}. 
\]
This number is interpreted as the {\em average amount of new infections an infected person causes during one infectious period if introduced in a completely susceptible population}. In standard compartmental models (see  \cite{van2002reproduction}), when  $\mathcal{R}_0 > 1$, there is a locally asymptotically stable endemic equilibrium. Conversely, when $\mathcal{R}_0 < 1$ the disease goes extinct, with the infected state variables converging to zero. If we consider the {\em time dependent} reproduction number 
\begin{equation}\label{mathcalR_time_dependent}
{\mathcal{R}(t) \coloneqq \sqrt{\frac{\beta_{hv} \beta_{vh} M_{\rm per}(t)}{\kappa \mu_M \ }}
},
\end{equation}
which depends on the mosquito population at time $t$, our conditions \eqref{condi} and \eqref{condii} of Theorem \ref{Theo 3.3} would amount, respectively, to 
$$\max_{t\in[0,\tau)} \mathcal{R}(t) \leq 1$$ 
for convergence to $0,$ and to 
$$ \min_{t\in[0,\tau)} \mathcal{R}(t) > 1$$ 
for the disease to be endemic, with a globally attractive strictly positive periodic solution. 

\section{Numerical simulations}\label{Ch2:NumSimSection}

In this last section we present simulations to illustrate our model and results. In Table \ref{Parameter_table} we expose ranges of realistic values for the  parameters involved in system \eqref{sec2:sys3}. In Subsection \ref{scenario simulation section}, we present some scenarios for the possible asymptotic behaviours of the system, as predicted by Theorems \ref{Theo 3.2} and \ref{Theo 3.3}. Subsection \ref{analysis ssec} explores the conditions from Theorems \ref{Theo 3.2} and \ref{Theo 3.3}, analyzing the effect of the system's parameters on its asymptotic behaviour.

\subsection{Parameters and Survey}\label{parameter section}

The locality of M\^ancio Lima, in the state of Acre in northwestern Brazil, is an important malaria hotspot. There, malaria is strongly associated with fish farming \cite{Reis2015malj, Reis2015plos}. 
In this area, there are ponds with commercial fishes and ponds with natural fishes, {as well as }natural water bodies (creeks, swamps) {were {\em Anopheles darlingi} and other malaria vectors breed}.
In Table \ref{Parameter_table} we gathered realistic parameters for the situation in M\^ancio Lima. Some data was gathered in a visit to the region in 2018.

\begin{table}[ht]
    \centering
    \label{Parameter_table}
\resizebox{0.95\textwidth}{!}{
\begin{tabular}{cccccc}
 \hline
  {\bf Parameter} & \begin{tabular}{@{}c@{}}{\bf Value/} \\ {\bf Range}\end{tabular} & \begin{tabular}{@{}c@{}}{\bf Note/} \\ {\bf Ref.}\end{tabular} & {\bf Scenario 1} & {\bf Scenario 2} & {\bf Scenario 3} \\
 \hline
 \hline
    $r$    & $0.01666 \text{ day}^{-1} $ &\ref{commentr}  &  0.01666 & 0.01666 & 0.01666  \\
 
    $\tau$  & $30-60 \text{ days}$& \ref{commentr}  & 30 & 30 & 30  \\

    $\gamma(t)$  & $0.5-1\text{ dimensionless}$& \ref{commentr},\cite{Reis2015plos,Reis2015malj} & 0.65 & 0.65 & 0.65 \\
    
    $\alpha$  & $8.75 - 43.66 \text{ day}^{-1}$& \cite{Santos1981,robert1983field}& {43.66} & {43.66}  &  {8.75}  \\
    
    $\nu$  & $1/(15.6 \pm 2.86)\text{ day}^{-1}$& \cite{Santos1981}& {0.0641} & {0.0641} & {0.0541}  \\
    
    $\mu_L$  & $0.51 - 0.79\text{ day}^{-1}$& \cite{villarreal2015establishment}& {0.62}  & {0.62}  &{0.99} \\ 
    
    $\mu_p$  & $0.11 - 0.39\text{ day}^{-1}$ &\ref{commentmup} & 0.31 & 0.31 & 0.31 \\
    
    $\mu_M$  & $0.089-0.476 \text{ day}^{-1}$ &\cite{Barros2011}& {0.089} & {0.16} & {0.8}  \\
    
    $\kappa$  & $0.024-0.16  \text{ day}^{-1}$ &\ref{commentkappa} & 0.05   & 0.05   & 0.05 \\
    
    $\beta_{vh}$  & $0.02 - 0.25 \text{ day}^{-1}$& \cite{mandal2011mathematical}& {0.2} & {0.1} &  {0.2}  \\
    
    $\beta_{hv}$  & $ 0.05 - 0.25\text{ day}^{-1}$& \cite{mandal2011mathematical}& {0.2} &  {0.1} &  {0.2} \\
    
 \hline
\end{tabular}}
\caption{Biologically feasible parameter ranges.}
\end{table}

\bigskip

\noindent{\bf Comments about Table \ref{Parameter_table}:}

\begin{enumerate}
\item\label{commentr} From interviews with the fish farmers, it was inferred that it takes two months for the vegetation to completely cover the {borders of} fish {ponds} and border cleaning {activities are performed every} 30  to 60 days. Moreover, at each cleaning episode more than half of the border vegetation is removed, which makes the cleaned proportion  $\gamma(t)$ greater than $0.5$.

%\item\label{commentgamma} %\deleted{The median amount of larvae in ponds with  less than $20\%$ of vegetated border was 1300 larvae/m, while a pond with more than $80\%$ vegetated border had 3700 larvae/m %\cite{Reis2015plos,Reis2015malj}. This gives a gross estimate of an average efficacy for the effect of cleaning on larval abundance of 0.65.} \cla{Verificar, n\~ao est\'a muito claro dada a defini\c{c}\~ao de $\gamma$ que \'e uma fun\c{c}\~ao. Explicar melhor.} \cla{Put confidence interval from our original article.}
%\item\label{commentalpha} Range calculated from the mean number of hatched eggs per oviposition $(83 \pm 48)$ \cite{Santos1981}, divided by the length of the gonotropic cycle (3 - 4 days) \cite{robert1983field}.
\item\label{commentmup} Pond vegetation interferes with the feeding behavior of the fish, and indirectly protects the larvae from their predators \cite{Kamareddine2012}. Without vegetation, larvivorous fishes can reduce the amount of larvae in 90\%. 

\item\label{commentkappa} The range of values for $\kappa$ is calculated from \cite{Aguas2011}. Adherence to {malaria} treatment directly affects the recovery rate \cite{pereira2011adherence}. We choose a value of $\kappa$ for the simulation corresponding roughly to 14 days of treatment and a week until symptoms onset.

\end{enumerate}

\subsection{\sol{Numerical Illustration of the Main Results}}\label{scenario simulation section}
    We aim at illustrating Theorems \ref{Theo 3.2} and \ref{Theo 3.3} with numerical simulations. The following scenarios will be simulated.

\noindent\hspace{6pt}{\bf Scenario 1:} Positive periodic mosquito and  infected human populations.

\noindent\hspace{6pt}{\bf Scenario 2:} Positive periodic mosquito population, and infected population converging to $0$.

\noindent\hspace{6pt}{\bf Scenario 3:} Mosquito population and infected population converging to $0$.

The numerical values used for the simulations are given in Table  \ref{Parameter_table}. For $K_w(V)$ and $K_p(V)$, we choose the following arbitrary functions: 
\begin{equation}
\label{K_p_and_K_w}
\begin{split}
K_p (V) \coloneqq \frac{0.8}{1 + e^{-5(V - 0.5)}},\quad K_w (V) \coloneqq 4 V + 0.5 .
\end{split}    
\end{equation}
The equation for $K_w(V)$ was interpolated by taking $K_w(0.2) = 1.3$ thousand larvae and $K_w(0.8) = 3.7$ thousand larvae.  $K_p$ is chosen to be a sigmoid function with arbitrary parameters such that $K_p(0.5) = 0.4$ thousand larvae. \fel{While the numerical parameters were arbitrary, the choice of a linear function for $K_w$ and a sigmoidal function for $K_p$ was made by extrapolating the results presented in \cite[Fig. 5]{Reis2015malj}.} We chose two different kinds of functions so as not to simplify the conditions on Theorem \ref{Theo 3.2}. We consider that the vegetation has assumed periodic behaviour as described in Proposition \ref{ch2:sec3:Vper_sol}, with the asymptotic value $\gamma^* = 0.65$. The values for $\mu_M$ and $\mu_L$ in Scenario 3 were chosen outside of the biological feasible range in order to force the differential equation system to assume the desired asymptotic behaviour.
   System \eqref{sec2:sys3} was integrated numerically for each scenario fixing the initial conditions
   $$
   (V,S,I,M_{S},M_{I},L_{p},L_{w})(0)= (0.7,0.9,0.1,0.5,0,2,0.2).
   $$
   The values for the conditional expressions of both Theorem \ref{Theo 3.2} and \ref{Theo 3.3} are given in Table \ref{Table:Conditions}. The simulations are shown in Figure \ref{Sec5.2Fig1}. We could then validate numerically the results predicted by Theorems \ref{Theo 3.2} and \ref{Theo 3.3}.

\begin{table}[tb]
\resizebox{\linewidth}{!}{
\begin{tabular}{cccc}
 \hline
  {\bf Expression} & {\bf 1} & {\bf 2} & {\bf 3}\\
 \hline
 \hline
 \\
 $\mathcal{N}$ &  45.967 & 25.569 & 0.567 \\ \\ 
\sol{$ \mathcal{N} \max_{t\in [0,\tau)} {\frac{{\tilde K}_w}{ {\tilde K}_w + {\tilde K}_p}} + \mathcal{N}_p(\tau^-) \max_{t\in [0,\tau)}{\frac{{\tilde K}_p}{ {\tilde K}_w + {\tilde K}_p}}$}
 &  \fel{49.579}  & \fel{27.487} & \fel{0.609} \\ \\
\sol{
 $\frac{\min_{[0,\tau)}{\left( \frac{ {\tilde K}_w (t)}{ {\tilde K}_w(t) + {\tilde K}_p (t)} \right)}}{\mathcal{N}^{-1} + \frac{ {\tilde K}_w (\tau^-) + {\tilde K}_p(\tau^-) }{{\tilde K}_w(0) + {\tilde K}_p(0)}} + 
 \frac{\min_{[0,\tau)} \frac{ {\tilde K}_p (t)}{ {\tilde K}_w(t) + {\tilde K}_p (t)}}{\mathcal{N}_p(0)^{-1} + \frac{ {\tilde K}_w (\tau^-) + {\tilde K}_p(\tau^-) }{{\tilde K}_w(0) + {\tilde K}_p(0)}} $
}
& \fel{0.002} & \fel{0.003} &  \fel{0.004} \\ \\ 
 \hline \\ 
    %$\displaystyle{\max_{t \in [0,\tau)}\sqrt{\frac{\beta_{vh}\beta_{hv}M_{\rm per}(t)}{\kappa \mu_M}} = 
    $\displaystyle{\max_{t \in [0,\tau)} \mathcal{R}(t)}  $ 
    & {3.899}  & {1.018} & \fel{0.000}\\ \\
    %$\displaystyle\fel{\frac{\beta_{vh}\beta_{hv}\min_{t \in [0,\tau)}M_{\rm per}(t)}{\kappa \mu_M}}$ 
     $\displaystyle{\min_{t \in [0,\tau)} \mathcal{R}(t)}  $
     & {3.288}  & {0.835} & \fel{0.000}\\ \\
 \hline
\end{tabular}}
\caption{Calculated values for the conditional expression.}
\label{Table:Conditions}
\end{table}

\begin{figure}[ht] 
\resizebox{\linewidth}{!}{
{\centering\includegraphics[width=0.75\linewidth]{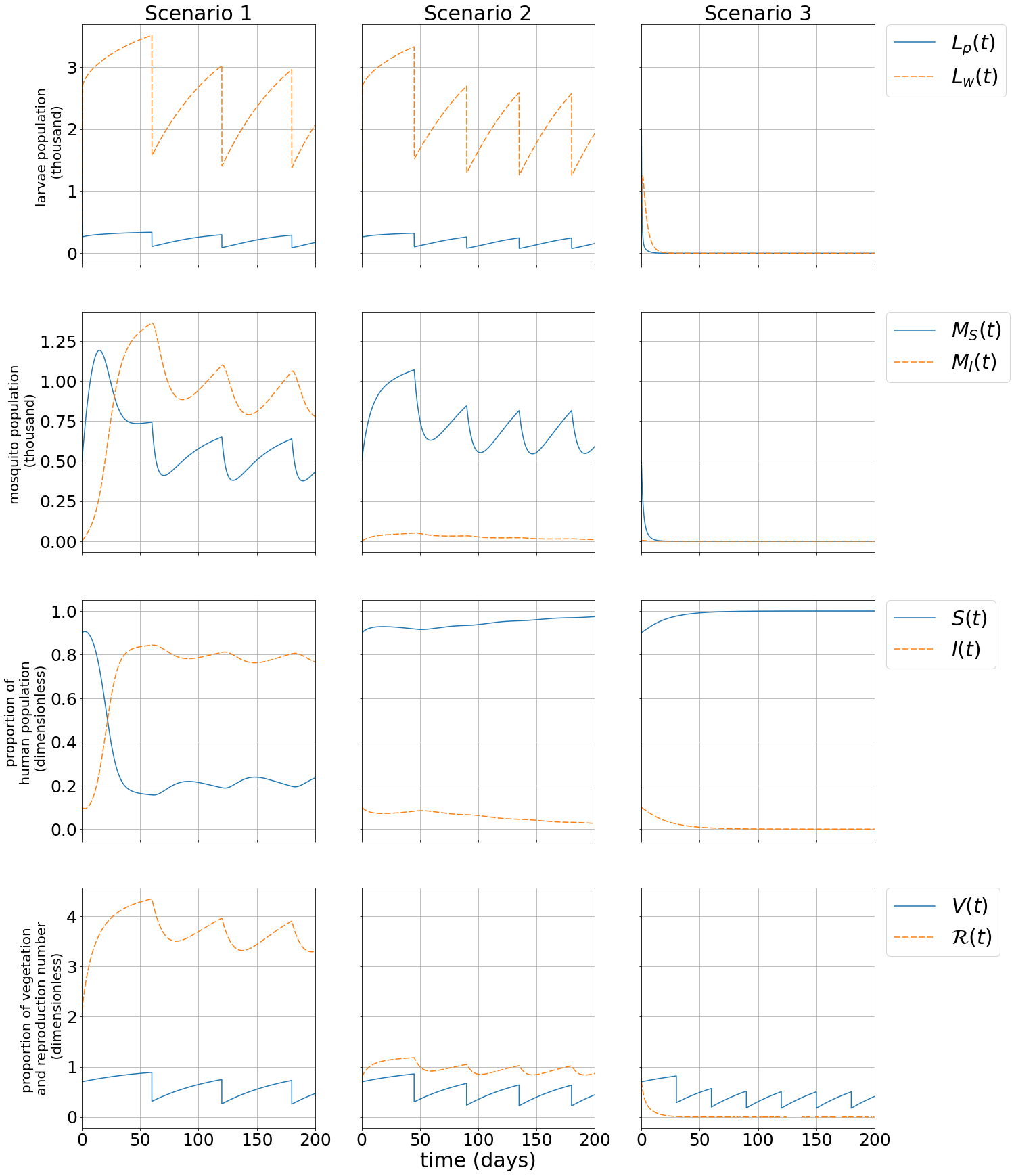}}\\
 } 
 \caption{Behaviour of the system under each scenario. Both mosquitoes and larvae are in the scale of thousands.% \sol{XXX: os gr\'aficos da \'ultima linha est\~ao mais ou menos. Os segundo e terceiro da direita s\~ao ileg\'iveis nessa escala.}
 }
 \label{Sec5.2Fig1}
\end{figure}    

%\begin{figure}[h!] 
%\centering
% \noindent{\centering\includegraphics[width=0.75\linewidth]{Fig5.png}}\\
%  \indent 
% \caption{Behaviour of the system under Scenario 2.}
% \label{Sec5.2Fig2}
%\end{figure}  

%\begin{figure}[h!] 
%\centering
% \noindent{\centering\includegraphics[width=0.75\linewidth]{Fig6.png}}\\
 % \indent 
 %\caption{Behaviour of the system under Scenario 3.}
 %\label{Sec5.2Fig3}
%\end{figure}  

\subsection{Analysis}\label{analysis ssec}

Using Table \ref{Parameter_table}, we can calculate the range of biologically feasible values for the basic offspring number $\mathcal{N}$, that gives $[0.953,55.126]$. It is possible to observe that elimination of the mosquito population through border vegetation removal is unfeasible if $\mathcal{N}$ is in the range of biologically feasible parameters. That is the reason why, in order to simulate a situation like the one of Scenario 3, we chose values for $\mu_M$ and $\mu_L$ that are outside the range observed in nature \cite{Barros2011,villarreal2015establishment}. 

Now, consider Theorem \ref{Theo 3.3}. In situation 1 of that theorem, the disease is eradicated. In situation 2, it becomes endemic, with varying incidence levels. Again, both conditions are {\em sufficient} conditions, and do not exhaust the possibilities. The case where $$
{\min_{t \in [0,\tau)} \mathcal{R}(t)} \leq 1 < {\max_{t \in [0,\tau)} \mathcal{R}(t)},$$
was observed through numerical simulation to assume periodic behavior close to $0$. 
%Figure \ref{modeling:TheoM_illustrated} shows the possible situations described in Theorem \ref{Theo 3.3}.
%\begin{figure}
%\begin{center}
%\indent{\centering\includegraphics[width=1\linewidth]{Fig7.png}}\\
%  \indent 
%\caption[Three possible situations of Theorem \ref{Theo 3.3}.]{Three possible situations of Theorem \ref{Theo 3.3}. On the left, $M_{\rm per}$. On the right, the corresponding $I$ (in shades of blue) in the respective three situations: asymptotic periodic, indeterminate and convergent to $0$.}
%\label{modeling:TheoM_illustrated}
%\end{center}
%\end{figure}

From Theorem \ref{Theo 3.3}, the value $ {\left(\dfrac{\kappa \mu_M}{\beta_{hv} \beta_{vh}}\right)^2}$ is a threshold value for the mosquito population. If the maximum value of the mosquito population is below this threshold, the disease is eradicated. 
%Again, Table \ref{Parameter_table} allows us to calculate the biologically feasible range for it, giving $\cancel{[0.114, 29.260]} {[0.0002,856.1476]}$. 

%The system's evolution without cleaning is drawn as control, and the values for $\tau$ range in the set $\{60,30,20,15,10,5\}$. In Figure \ref{Sec5.3Fig1}, we represent the maximum and minimum value of both the $\mathcal{R}(t)$ and the infected (after they assume periodic behaviour) as a function of $\tau$. The green line at 1 shows the threshold value for $\mathcal{R}(t)$ 
    
\begin{table}
\centering
\resizebox{0.3\linewidth}{!}{
\begin{tabular}{lr}
 \hline
  {\bf Parameter} &  {\bf Values}\\
 \hline
 \hline
    $r$   &  0.01666  \\

    $\tau$ &  {\em varying}\\

    $\gamma$ &  0.65\\
    
    $\alpha$ &  26.2\\
    
    $\nu$ &  0.0641\\
    
    $\mu_L$ &  0.62\\ 
    
    $\mu_p$ &  0.31\\
    
    $\mu_M$ & 0.16\\
    
    $\kappa$  &  0.05 \\
    
    $\beta_{vh}$  & 0.2 \\
    
    $\beta_{hv}$ &   0.2\\
 \hline
\end{tabular}}
\caption{Value of parameters for the simulation of the effect of cleaning frequency change. Biological meaning of parameters is given in Table \ref{model:table1}}
\label{Table 5}
\end{table}
The biological meaning of the parameters in the conditions of Theorem \ref{Theo 3.3} are the  transmission rates, the  recovery rate, the  mosquito mortality and the size of the mosquito population. These give hints to the more effective ways of fighting malaria: prevention of bites, treating infected individuals quickly, and controlling mosquito population. From the epidemiologist or the public health practitioner point of view, the most important issue is to have a grasp on the relative effects of each factor on malaria transmission \cite{mandal2011mathematical}. 

We simulate the effect of an increase in the frequency of cleaning, using the set of parameters described in Table \ref{Table 5}. The values for $\tau$ range from 70 to 5 days. In Figure \ref{Sec5.3Fig1}, we represent the maximum and minimum values of both $\mathcal{R}(t)$ (as defined in \eqref{mathcalR_time_dependent}) and the size of the infected human population (\sol{once} they assume periodic behaviour) as a function of $\tau$. \sol{In the graphic on the left hand-side, t}he green \sol{horizontal} line at 
1 shows the threshold value for $\mathcal{R}(t)$. 

The effect of an increased frequency of cleaning is a reduction of the average mosquito population, and consequently a reduction of the average infected population. Eventually, the solution's asymptotic behaviour changes from asymptotically positive periodic to convergence towards $0$. Even with moderate efforts, the average proportion of infected humans is significantly reduced.
\begin{figure}[ht] 
\centering
 \noindent{\centering\includegraphics[width=1\linewidth]{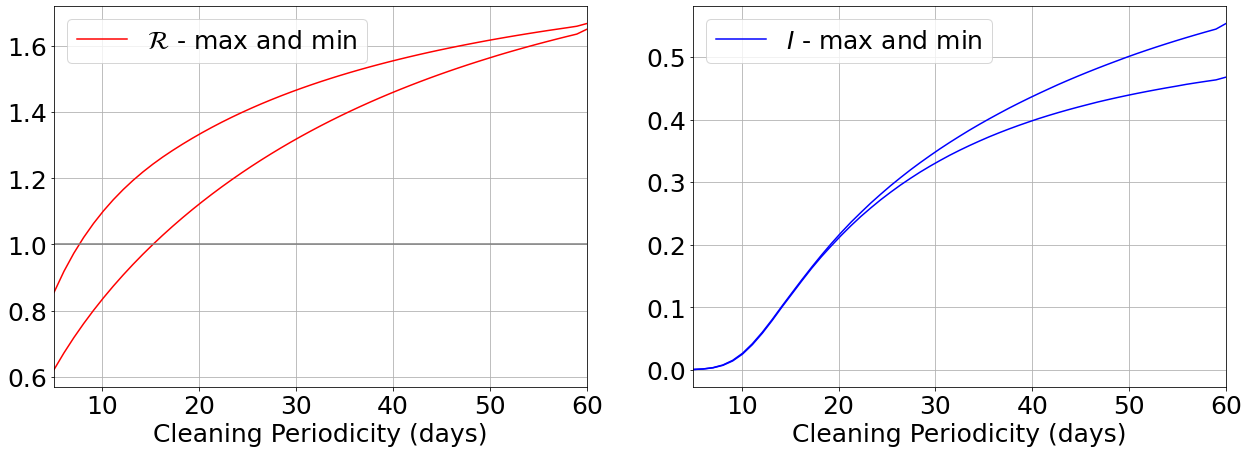}}\\
  \indent 
 \caption[Behaviour of the disease under \sol{different} cleaning frequencies.]{Behaviour of the disease under  \sol{different} cleaning frequencies.
On the left, the lines represent the minimum and maximum values of $\mathcal{R}(t)$ as a function of the cleaning periodicity.  
 On the right, the minimum and maximum values for the infected population $I.$
 It is possible to observe that as $\max \mathcal{R}(t)$ becomes smaller than $1$, the proportion of infected approaches $0$.}
 \label{Sec5.3Fig1}
\end{figure}  

%\subsection{The effect of awareness in cleaning periodicity}\label{awareness ssec} \fel{I think this section, too, can be omitted. Although it is interesting, I think it deviates too much from our main subject.} \sol{I agree, we can put it in another article, of different nature.}

\section{Summary and Conclusion}
 
    %\remove{In this paper, we use field observations and data about malaria in Acre, state in the Northwest of Brazil, to propose a system of differential equations to model the interaction between the border vegetation of fish ponds, the life cycle of {\em Anopheles darlingi} and the transmission of malaria. The model considers the differences observed in commercially active and inactive ponds, and the impact of maintenance on {\em Anopheles} larvae density. As far as we know, this is the first model linking malaria and fish farming.}
    
     %  \remove{ In our model, we assumed that fish pond cleaning happened periodically and disregarded spatial effects on the dynamics. The model could potentially be improved by dropping these assumptions. Another important development is the collection of data to asses the validity of the model.}
     
   %  \remove{Through Theorems \ref{Theo 3.2} and \ref{Theo 3.3}, we established conditions for determining the asymptotic behaviour of the mosquito population and the incidence of disease. Numerical simulations of our model have shown that an increase in cleaning frequency can reduce the mosquito population, which in turn reduces incidence of malaria. Our results suggests that incentives for fish pond maintenance are potential public health measures in the state of Acre.}

{Malaria causes a large disease burden in the tropic regions of the world. As climate changes, its geographical range is expanding to higher altitudes and latitudes. In the absence of vaccines,  incidence reduction is highly dependent on vector control either by the removal of breeding sites or by the use of insecticides. In different parts of the world, fish farming in malaria endemic areas has been stimulated, without considering the risks it poses to the population, when not properly managed \cite{howard2008abandoning, maheu2010risk}. }  

 {In this paper, we used field observations and data from a fish farming area in Acre, Northwest of Brazil, to propose a system of differential equations to model the interaction between the rate of border vegetation cleaning of fish ponds, the life cycle of \textit{Anopheles darlingi} and the transmission of malaria. The model considers the differences observed between vector populations in commercially active and inactive ponds. As far as we know, this is the first model describing the links between malaria and fish farming. The model is used to study the impact of the periodic cleaning of ponds on malaria incidence. Through Theorems \ref{Theo 3.2} and \ref{Theo 3.3}, we established conditions for determining the asymptotic behaviour of the mosquito population and the incidence of disease. Numerical simulations show that an increase in cleaning frequency can reduce the mosquito population, which in turn reduces incidence of malaria. This result highlights the importance of continuous maintenance of ponds. This can be a challenge since cleaning is often costly and requires high effort. When fish farmers were asked about these practices, they considered them hard to achieve since malaria prevention is often not seen as a sufficient driver for action, compared with other problems; and continuous actions require equipment that are not available \cite{alves_codeco_peiter_souza-santos_2019}.}       

{The biology of malaria vectors is one of the best known. Still, to model this system, some assumptions were made without strong evidence from the literature. For example, the effect of vegetation on recruitment rate was modelled as a logistic function, where $K$ depends on the vegetation border. More studies on the effect of vegetation on density dependence would be helpful to better parameterize this relationship. Moreover, we disregarded spatial effects on the vector dynamics. The region has hundreds of ponds of different sizes. Here we pooled them together as a single pond with fish and a single pond without fishes. The model could potentially be improved by dropping these assumptions. }
	 
{The proposed model can have other applications, for example, combining cleaning practices with other control strategies, such as periodic use of insecticides, which has been tested in fish ponds \cite{fontoura2021monthly}.}

\section*{Acknowledgments}
This research was funded by the Applied Research Project of FGV - Fundação Getúlio Vargas entitled {\em ``Modelagem, análise e estimativa da contribuição dos tanques de piscicultura na população do mosquito anopheles e o impacto na transmissão da malária no Alto Juruá, Acre''}, by the Program {\em Jovem Cientista do Nosso Estado} of FAPERJ, Brazil, by the STIC AmSud funding for the MOSTICAW Project (Process No. 99999.007551/2015-00) and by CNPq Grant No. 454665/2014-8.
The authors particularly acknowledge the fish farmers from M\^ancio Lima, Acre, Brazil, for their availability.

{We acknowledge the anonymous reviewers for their careful reading and comments that helped us significantly improve this manuscript.}
\appendix

\section{{Remarks on the parameter $K$}}
\label{AppK}

{
The parameter $K$ can be interpreted in two ways: as the reciprocal of the derivative of the {\em per-capita} recruitment rate with respect to $L$, or as the maximal value for $L$ that can be attained when $M_S + M_I$ goes to $  \infty$. 
Consider the equation (2.1e):
\[
    \dot L = \alpha (M_S + M_I)\left(1 - \frac{L}{K}\right) - (\nu + \mu_L)L.
\]
The adult rate of larvae production can be written as $\rho(L) = \alpha\left(1 - \frac{L}{K}\right)$, which can be interpreted as a recruitment rate with larval density dependence.
This recruitment rate has a maximum value $\rho(0) = \alpha$ and decreases linearly with slope $K^{-1}$.
This provides the first interpretation for the parameter $K$.
On the other hand, if $L(0) < K$, note that $K$ is an upper bound for larval abundance since $\dot L \leq 0$ for $L \geq K$.
However, it can be shown that lower upper bounds exist.
Assuming that $L(0) < K$, we can rewrite equation (2.1e) as:
\[
    \dot{L} 
     = \alpha(M_S+M_I)\left(\frac{1}{K}+\frac{\nu+\mu_L}{\alpha(M_S+M_I)}\right) \left(\frac{K}{1+\frac{K(\nu+\mu_L)}{\alpha(M_S+M_I)}}-L \right),
\]
from which we can deduce that
\[
L \leq  \frac{K}{1+\frac{K(\nu+\mu_L)}{\alpha(M_S+M_I)}} < K.
\]
Note that this tighter lower bound approaches $K$ from below as $M_S + M_I$ tends to infinity. 
}
 
\section{On Carath\'edory solutions}\label{Ch3:sec2}

Here we recall existence and uniqueness results for Cauchy problems with right-hand side functions that are merely measurable with respect to time. In Section \ref{poincare_subsec} we employ these results to prove that the {\em Poincar\'e map} associated to this type of system is well-defined, continuous and differentiable.

\vspace{5pt}

\noindent{\bf Basic Assumptions.} For an open set $\Omega \subset \mathbb{R} \times \mathbb{R}^n$  and a function $F \colon \Omega \mapsto \mathbb{R}^n$, let us consider the conditions \eqref{sec2:Ass1} and \eqref{sec2:Ass2} below.
\begin{align}\tag{A}\label{sec2:Ass1}
\begin{cases}
    t \mapsto F(t,x)\ \text{is measurable on } \Omega_x = \{t : (t,x) \in \Omega\}, \text{ for all } x, \\
    x \mapsto F(t,x)\ \text{is continuous on }\ \Omega_t = \{x : (t,x) \in \Omega\}, \text{ for a.e. } t.
\end{cases}
\end{align}
\begin{align}\tag{B}\label{sec2:Ass2}
\begin{cases}
\text{For any compact set $K \subset \Omega$, there exist $C_K>0, L_K>0$ such that}\\
|F(t,x)| \leq C_K,\,\, |F(t,x) - F(t,y)| \leq L_K|x - y|,\quad \text{for all}\ (t,x), (t,y) \in\ K.
\end{cases}
\end{align}

\vspace{5pt}

We begin by recalling the following result \cite{rudin2006real,bressan2007introduction,coddington1955theory}.

\begin{theorem}[Carath\'eodory's Existence Theorem]\label{sec3:theo1}
    Given $F \colon \Omega \to \mathbb{R}^n$ satisfying \eqref{sec2:Ass1} and \eqref{sec2:Ass2}, and some $(t_0,x_0) \in \Omega$, consider the Cauchy problem
\begin{equation}\label{sec3:theo1:eq1}
    \dot x = F(t,x),\quad x(t_0) = x_0.
\end{equation}
     The following assertions hold.
\begin{itemize}
    \item[(i)] There exists $\epsilon > 0$ such that \eqref{sec3:theo1:eq1} has a local solution defined on $ [t_0, t_0 + \epsilon].$
    \item[(ii)] Moreover, if $\Omega = \mathbb{R}\times \mathbb{R}^n$ and there exist constants $C,L$ such that
\begin{equation}\label{sec3:theo1:con1}
    |F(t,x)| \leq C,\ |F(t,x) - F(t,y)| \leq L|x - y|,\ \ \text{for all}\ \ (t,x),(t,y) \in \mathbb{R}\times\mathbb{R}^n,
\end{equation}
    then, for every $t_1>t_0$, the initial value problem \eqref{sec3:theo1:eq1} has a unique global solution defined on $ [t_0,t_1]$. Moreover, the solution depends continuously on the initial data $x_0$.
\end{itemize}
\end{theorem}

\begin{theorem}[Uniqueness of solutions]\label{sec3:theo2.5}
Under the hypotheses of Theorem \ref{sec3:theo1}, let $x_1$ and $x_2$ be solutions of \eqref{sec3:theo1:eq1} defined on the intervals $[t_0,t_1]$, $[t_0,t_2],$ respectively. If $t' \coloneqq \min\{t_1,t_2\}$, then $x_1(t) = x_2(t),$ for $t \in [t_0,t']$.
\end{theorem}

\begin{theorem}[Maximal solutions]\label{sec3:theo3}
Under the hypotheses of Theorem \ref{sec3:theo1}, let $t^* > t_0$ be the supremum of all times $t_1$ for which \eqref{sec3:theo1:eq1} has a solution $x$ defined on $[t_0,t_1]$. Then, either $t^* = \infty$ or
\[    
\lim_{t \to t^*_-} \left(|x(t)| + \frac{1}{d\big((t,x(t)),\partial \Omega\big)}\right) = \infty.
\]
\end{theorem}

\section{Cooperative systems with time-measurable dynamics and concave nonlinearities}\label{appendix_cooperative}

    In this section we show properties of the {\em Poincar\'e map} associated to system \eqref{sec3:theo1:eq1}. These properties are used to extend  \sol{Jiang}'s result \cite[Theorem 5.5]{jifa1993cooperative} to time-measurable differential equations.

    %add internal references to the main theorem of each section on the outline.

\subsection{Well-Definedness and Differentiability of Poincar\'e Map}\label{poincare_subsec}

Let us introduce the following notation: for given $(t_0,x_0)$ in $ \Omega,$ we write ${\bf x}(t;t_0,x_0)$ to denote the solution of \eqref{sec3:theo1:eq1} at time $t,$ whenever it exists.
For fixed $t > t_0,$ we define the {\em Poincar\'e map} $T$ corresponding to system \eqref{sec3:theo1:eq1} as the function that to each $x_0$ with  $(t_0,x_0) \in \Omega,$ associates ${\bf x}(t;t_0,x_0),$ {that is,} $T(x_0) \coloneqq {\bf x}(t;t_0,x_0).$

    Naturally, for the Poincaré map to be well-defined, we require that  for any given $x_0,$ the underlying Cauchy problem has a unique solution. In the case the system admits a forward invariant \sol{set,} we can define the Poincar\'e map on this set for any $t>t_0$, as stated in the following corollary.

\begin{corollary}[Well-definiteness of the Poincar\'e map]\label{sec3:Cor1}
    Assume that the hypotheses of Theorem \ref{sec3:theo1} hold and  that the differential equation $\dot x = F(t,x)$ admits a compact forward invariant set $\mathcal{C},$ with $[t_0,\infty)\times \mathcal{C} \subseteq \Omega.$  Then, the associated Poncar\'e map is well-defined on $\mathcal{C}$. %\sol{Qual seria o conjunto $\mathcal{C}$ em nossos exemplos? Esse \'e um ponto importante, dado que s\~ao teorema de ponto fixo. Acho que isto d\'a problema.}
\end{corollary}

\begin{proof} 
It follows straightforwardly from Theorems \ref{sec3:theo1}-\ref{sec3:theo3}.
\end{proof}

    The following Theorem proves differentiability of the Poincaré map w.r.t. the initial conditions.

\begin{theorem}[Differentiability of the Poincar\'e Map]\label{sec3:theo2}
    Suppose that the hypotheses of Corollary \ref{sec3:Cor1} hold, and assume
    further that $g$ is continuously differentiable with respect to $x$. 
    Let $t_0\in \mathbb{R},$  $t>t_0$ and consider the Poincar\'e map $T$ associated to the Cauchy problem \eqref{sec3:theo1:eq1} and time $t$. Then,  the map $x_0 \mapsto T  x_0 \coloneqq {\bf x}(t;t_0,x_0)$ is continuously differentiable with respect to $x_0$. Its Jacobian matrix is
\begin{equation}
    D_{x_0}T  = M(t,t_0),
\end{equation}
    where $M(\cdot,\cdot)$ is the fundamental matrix associated to the linear problem
\begin{equation}
    \dot v (t) = D_x F(t,{\bf x}(t;t_0,x_0))\cdot v (t).
\end{equation}
\end{theorem}
\begin{proof}
It follows from \cite[Theorem 2.10]{bressan2007introduction}.
\end{proof}

\subsection{Asymptotic behaviour}

In view of the properties of the Poincar\'e map stated  above, we can extend  \cite[Theorem 5.5]{jifa1993cooperative} to the differential equation
\begin{equation}
\label{ch3:sec1:eq1}
\dot x = \mathcal{F}(t,x), \quad (t,x) \in \mathbb{R}_+ \times [0,p] ,
\end{equation}
where $p$ belongs to $\mathbb{R}^n_+$ and $\mathcal{F}$ is \sol{continuous with respect to $x$ and measurable with respect to $t$ in its domain,} \fel{$\tau$-periodic in $t$. Moreover, $D_x \mathcal{F}$ exists and is continuous with respect to $x$.} \sol{Let us consider the following hypotheses on $\mathcal{F}$:}

\begin{align}\tag{I}\label{HypothesisI}
\begin{cases}
   \fel{\text{if $x\geq 0$ with $x_i = 0$, then } \mathcal{F}_i (t,x) \geq 0, \text{ for $1\leq i \leq n$ } },\\
   \fel{\text{if $x\leq p$ with $x_i = p_i$, then } \mathcal{F}_i (t,x) \leq 0, \text{ for $1\leq i \leq n$ } }.
\end{cases}
\end{align}

\begin{align}\tag{M}\label{HypothesisM}
\begin{cases}
   \fel{D_x\mathcal{F}(t,0) \text{ is irreducible for every } t,}\\
   \fel{D_x \mathcal{F}(t,x)\text{ has non-negative off diagonal terms, for } (t,x) \in \mathbb{R}_+ \times [0,p]. }\\
\end{cases}
\end{align}

\begin{align}\tag{C}\label{HypothesisC}
\begin{cases}
   \fel{D_x \mathcal{F}(t,x) >_{\neq} D_x\mathcal{F}(t,y),\quad \text{if } 0<x<y<p.}
\end{cases}
\end{align}

\sol{At this point, let us recall the notation $\overline A$ and $\underline A$ introduced in the notation Section \ref{SecNotation}.}

\begin{theorem}[Conditions for Periodic Solutions]\label{Ch3:Main_Theo}
    Let $$\mathcal{F}:  \mathbb{R} \times [0,p]\ \to \mathbb{R}^n_+$$
    where $p\in \mathbb{R}^n_+$ be a function verifying \eqref{HypothesisI}, \eqref{HypothesisM} and \eqref{HypothesisC}.  Assume further that $\mathcal{F}(t,0)\equiv 0$ and all  the trajectories of  \eqref{ch3:sec1:eq1} are bounded.
Then, setting $A(t):=D_x \mathcal{F}(t,0)$, the following conditions hold.
\begin{enumerate}
   \item If all principal minors of $- \overline A$ are non-negative, then $lim_{t \to \infty} x(t) = 0$ for every solution of \eqref{ch3:sec1:eq1} with initial condition in $[0,p]$.
   \item If $- \munderbar{A}$ has at least one negative principal minor, then \eqref{ch3:sec1:eq1} possesses a unique $\tau$-periodic solution which attracts all initial conditions in $[0,p]\backslash \{0\}$.
\end{enumerate}
\end{theorem}
\begin{remark}
\label{RemThms}
    \fel{In the proof of Theorem \ref{Ch3:Main_Theo} (see \cite{smith1986cooperative,jifa1993cooperative}), the only property of $\mathbb{R} \times [0,p]$ that is used is the fact that for any $t\in\mathbb{R}$ it is bounded over the state space. This removes the possibility of unbounded trajectories. Therefore, we can safely extend the result for sets of the form $\{ (t,x): t\in\mathbb{R}, 0\leq x \leq p(t) \}$ for $p: \mathbb{R} \mapsto \mathbb{R}^n_+$} 
\end{remark}

\section{Impulsive equations}
\label{appendix_impulsive}

\begin{lemma}\label{Ch2:sec3:lem1}
For a sequence  $(\gamma_n)_{n\in\mathbb{N}}\subset \mathbb{R}$ and  $u_0 \in \mathbb{R}$,  consider the initial value problem 
\begin{equation}
\label{ch2:sec3:lem1:eq1}
\begin{split}
    \dfrac{d u}{d t}  &=  r(1 - u(t)),\quad \text{for}\ t \neq n\tau,\\
    u(n \tau^+) &= (1 - \gamma_n) u(n \tau^-),\quad \text{for } n \in \mathbb{N},\\
    u(0) &= u_{0}.
\end{split}
\end{equation}
  Then, the following assertions hold.
  \begin{itemize}
      \item[(i)] If $\gamma_n=\gamma\in [0,1]$ for all $n\in \mathbb{N}$ and $u_0 =\displaystyle 1 - \frac{\gamma}{1 - (1 - \gamma )e^{-r\tau}}$, then \eqref{ch2:sec3:lem1:eq1} admits the periodic solution $u_{\rm per}$ given by
\begin{align}\label{ch2:sec3:lem1:per_sol}
    u_{\rm per}(t) &\coloneqq 1 - \frac{\gamma e^{-r(t - n\tau)}}{1 - (1 - \gamma )e^{-r\tau}},\quad  \text{for}\ t \in [n\tau, (n+1)\tau).
\end{align} 
    \item[(ii)] Moreover, if $(\gamma_n)\subset [0,1]$ is any sequence converging to  $\gamma,$  then any solution $u$ of \eqref{ch2:sec3:lem1:eq1} for any initial condition $u_0\in [0,1]$ verifies
\[
    \max_{t \in [n\tau,(n+1)\tau)} |u(t) - u_{\rm per}(t)| \to 0 \ \text{as} \  n \to  \infty.
\]
  \end{itemize}
\end{lemma}

\begin{proof}
Item (i) follows easily.
Let us prove item (ii) assuming initially that  $\gamma_n = \gamma$ for all $n\in \mathbb{N}$. The proof consists in taking the sequence formed by the points of discontinuity and studying its convergence. 
We begin by calculating $u\big((n+1)\tau^+ \big)$ as a function of $u\big(n \tau^+ \big)$. By solving the initial value problem
\begin{equation*}
\begin{split}
    \dot u(t) &= r(1 - u(t)),\quad \text{for}\ t  \in (n\tau,(n+1)\tau),\quad  u(n\tau) = u(n\tau^+),
\end{split}
\end{equation*}
and taking the limit as $t \to (n+1)\tau^{-}$, we get:
\begin{equation}
    u\big((n+1)\tau^-\big) = 1 - [1 - u\big(n\tau +\big)]e^{-r\tau}. \label{impulsive_per_solution_lemma_aux1}
\end{equation}
On the other hand, we can write $u((n+1)\tau^+)$ as a function of $u((n+1)\tau^-)$, using the impulse at $(n+1)\tau$, as
$
    u\big((n+1)\tau^+\big) = (1-\gamma) u\big((n+1)\tau^-\big).
$
Using latter equation and \eqref{impulsive_per_solution_lemma_aux1} yields 
\begin{equation}
\label{Ch2:sec3:recurrence_equation}
    u((n+1)\tau^+) = (1 - \gamma)\Big(1 - \big(1 - u(n\tau^+)\big)e^{-r\tau}\Big),
\end{equation}
    defining a recurrence relation for $u(n\tau^+)$ induced by the function 
    \begin{equation}
        \label{f}
        x \mapsto  f(x) \coloneqq (1 - \gamma)\Big(1 - \big(1 - x \big)e^{-r\tau}\Big).
    \end{equation}
    As $f$ is a contraction over $[0,1],$ the sequence $\big(u(n\tau^+)\big)_{n\in \mathbb{N}}$ converges to its fixed point
\begin{equation}\label{fixed_point_eq}
     u^* = \frac{(1 - \gamma)(1 - e^{-r\tau})}{1 - (1 - \gamma)e^{-r\tau}}.
\end{equation}
    This means that any solution $u$ converges to $u_{\rm per}$ given by:
\begin{equation*}
\begin{split}
    u_{\rm per}(t) = 1 - [1 - u^*]e^{-r(t - n\tau)} = 1 - \frac{\gamma e^{-r(t + n\tau)}}{1 - (1 - \gamma)e^{-r\tau}},\quad t \in [n\tau, (n+1)\tau),\, \, n \in \mathbb{N}.\\ 
\end{split}
\end{equation*}
Finally, in the generic case, notice  that for each $\gamma_n$, we have a corresponding $f_n$. We have to show that the infinite composition 
$$ 
\lim_{n \to \infty} f_n \circ \cdots \circ f_1(x) 
$$ 
converges. As  $\gamma_n \to \gamma > 0$, from \eqref{f} we get that $f_n \to f$ uniformly on $[0,1]$. For every $n \in \mathbb{N}$, $f_n$ is a contraction, which implies it has a fixed point $u^*_n$ in the form of equation \eqref{fixed_point_eq}. From \eqref{fixed_point_eq}, we can see that $u^*_n \to u^*$. These conditions guarantee us (see \cite{lorentzen1990compositions}) that the infinite composition converges to the fixed point of its limit function $f$. Thus the proposition has been proven. 
\end{proof}
%https://en.wikipedia.org/wiki/Infinite_compositions_of_analytic_functions

\bibliographystyle{siamplain}
\bibliography{references}

\end{document}

%% file: ex_shared.tex
\usepackage{lipsum}
\usepackage{amsfonts}
\usepackage{graphicx}
\usepackage{epstopdf}
\usepackage{algorithmic}
\usepackage{tikz}
\usetikzlibrary{arrows}
\usepackage{mathtools,amssymb}

\ifpdf
  \DeclareGraphicsExtensions{.eps,.pdf,.png,.jpg}
\else
  \DeclareGraphicsExtensions{.eps}
\fi

\newsiamremark{remark}{Remark}
\newsiamremark{hypothesis}{Hypothesis}
\crefname{hypothesis}{Hypothesis}{Hypotheses}
\newsiamthm{claim}{Claim}

\title{Modeling and control of malaria dynamics\\  in fish farming regions
\thanks{To appear in SIAM Journal of Applied Dynamical Systems.
\funding{This research was funded by the Applied Research Project of Fundação Getúlio Vargas {\em ``Modelagem, análise e estimativa da contribuição dos tanques de piscicultura na população do mosquito anopheles e o impacto na transmissão da malária no Alto Juruá, Acre''}, by the Program {\em Jovem Cientista do Nosso Estado} of FAPERJ, Brazil, by the STIC AmSud funding for the MOSTICAW Project (Process No. 99999.007551/2015-
00) and by CNPq Grant No. 454665/2014-8.}}
}

\author{Felipe J.P. Antunes \thanks{Escola de Matem\'atica Aplicada FGV EMAp, Funda\c{c}\~ao Getulio Vargas, Rio de Janeiro, Brazil
(\email{fjpantunes2@gmail.com}).}
\and M. Soledad Aronna \thanks{Escola de Matem\'atica Aplicada FGV EMAp, Funda\c{c}\~ao Getulio Vargas, Rio de Janeiro, Brazil
(\email{soledad.aronna@fgv.br},
\url{https://sites.google.com/view/aronna}).}
\and Cl\'audia T. Code\c{c}o \thanks{Programa de Computa\c{c}\~ao Cient\'ifica-Funda\c{c}\~ao Oswaldo Cruz, Rio de Janeiro, Brazil 
  (\email{codeco@fiocruz.br}).}
}

\usepackage{amsopn}

\makeatletter
\newcommand*{\addFileDependency}[1]{% argument=file name and extension
  \typeout{(#1)}% latexmk will find this if $recorder=0 (however, in that case, it will ignore #1 if it is a .aux or .pdf file etc and it exists! if it doesn't exist, it will appear in the list of dependents regardless)
  \@addtofilelist{#1}% if you want it to appear in \listfiles, not really necessary and latexmk doesn't use this
  \IfFileExists{#1}{}{\typeout{No file #1.}}% latexmk will find this message if #1 doesn't exist (yet)
}
\makeatother

%% file: diagrama_s3.tex
\resizebox{0.95\textwidth}{!}{ %% to change size
\begin{tikzpicture}[
node distance = 3cm, 
%auto,
%>=latex', 
thick,
   box/.style = {fill=#1,
                draw, solid, thin, minimum size=9mm},
    xs/.style = {xshift=#1mm},
    baseline = 3.5cm
                    ]

\path[use as bounding box] (0.5,-1.5) rectangle (11.5,-3.5);

\path[->] 

node (h0) []{}
node (h0a) [above=9mm of h0]{}
node (h0b) [below=9mm of h0]{}

node (V) [box = white, right = 11 mm of h0] {$V$}

%(h0a) edge[] node[swap]{} (V.170)
%(V.190) edge[] node[swap]{} (h0b)

node (h1) [right=15mm of V] {}

node (Lw) [box = white, above=5mm of h1] {$L_w$}
node (Lp) [box = white, below=5mm of h1] {$L_p$}

(h0a) edge[] node (inLw) []{} (Lw)
(h0b) edge[] node(inLp) []{} (Lp)

(V.90) edge[dashed] (inLw)
(V.270) edge[dashed] (inLp)

node (aLw)[above = 12mm of Lw]{}
(Lw) edge node[auto]{} (aLw)

node (bLp)[below = 12mm of Lp]{}
(Lp) edge node (Lp_out) [auto, swap]{}  (bLp)
%(V) edge[dashed] (Lp_out) 

 node (h2) [right=30mm of h1] {}

  node (MS) [box = white, above=5mm of h2] {$M_S$}

 node (MI) [box = white, below=5mm of h2] {$M_I$}

  node (aMS) [above=12mm of MS] {}
 (MS) edge node[auto]{} (aMS)

 node (bMI) [below=12mm of MI] {}
 (MI) edge node[auto]{} (bMI)

 node (h3) [right=20mm of h2] {}

 node (S) [box = white, above=5mm of h3] {$S$}
 node (I) [box = white, below=5mm of h3] {$I$}

 (MS) edge node (betaI) [auto]{} (MI)

% (MS.190) edge[] node (MSLw) []{} (Lw.315)
 
 (Lw.0) edge node[above]{} (MS.170)

% (MS.215) edge[] node (MSLp) []{} (Lp.45)
 (Lp.0) edge[] node[below]{} (MS.190)

% (MI) edge[] node (MILp) []{} (Lp.45)
% (MI) edge[] node (MILw) []{} (Lw.315)

 (I.80) edge[->] node[right]{} (S.280)

 (S.260) edge[->] node[left] (betaMI) []{} (I.100)

(I) edge[dashed] (betaI)

(MI) edge[dashed] (betaMI)

;
\end{tikzpicture}}